\documentclass[conference]{IEEEtran}

\usepackage{tikz}
\usepackage{amsmath}

\usepackage{bm}
\usepackage{graphicx}
\usepackage{xcolor}
\usepackage{import}

\usepackage{hyperref}
\usepackage{paralist}

\usepackage{listings}
\usepackage[noend]{algpseudocode}
\usepackage{url}
\usepackage{multirow}
\usepackage{caption}
\usepackage{array}
\usepackage{rotating}
\usepackage{makecell}
\usepackage{colortbl}
\usepackage{enumerate}
\usepackage{booktabs}
\usepackage{paralist}
\usepackage[ruled,vlined]{algorithm2e}
\usepackage{amssymb}
\usepackage{commath}
\usepackage{amsthm}
\usepackage{bm}
\usepackage{framed}

\usepackage[binary-units=true]{siunitx}
\usepackage{adjustbox}

\renewcommand{\tablename}{Tab.}
\renewcommand{\figurename}{Fig.}

\renewcommand{\paragraph}[1]{\vspace{0.1cm}\noindent{\bf #1.}}

\algrenewcommand\algorithmiccomment[1]{\hfill \textcolor{gray}{$\triangleright$ \textit{#1}}}

\newtheorem{hypothesis}{Hypothesis}

\usepackage{paralist}

\usepackage[
	automake,
	nogroupskip, %
	nonumberlist, %
]{glossaries-extra}

\usepackage[
	nogroupskip, %
	nonumberlist, %
]{glossaries-extra}

\makeindex{}
\makeglossaries{}

\setabbreviationstyle[acronym]{long-short}

\newacronym[plural=DAGs]{dag}{DAG}{Directed Acyclic Graph}
\newacronym{pow}{PoW}{Proof of Work}

\newglossaryentry{RTS}
{
	name={\text{RTS}},
	description={
		Random Transaction Selection
	}
}
\newglossaryentry{block-creation-rate}
{
    name={\ensuremath{\lambda}},
    description={
            Block creation time
        }
}
\newglossaryentry{cnt-honest-miners}
{
    name={\ensuremath{\mathbb{H}}},
    description={
            The number of honest miners
        }
}
\newglossaryentry{cnt-malicious-miners}
{
    name={\ensuremath{\mathbb{G}}},
    description={
            The number of greedy miners
        }
}
\newglossaryentry{adversarial-mining-power}
{
    name={\ensuremath{\kappa}},
    description={
            Adversarial mining power
        }
}
\newglossaryentry{network-propagation-delay}
{
    name={\ensuremath{\tau}},
    description={
            Network propagation delay of blocks
        }
}
\newglossaryentry{gap-parameter}
{
    name={\ensuremath{c}},
    description={
            Gap parameter in PHANTOM
        }
}
\newglossaryentry{discount-function}
{
    name={\ensuremath{\gamma}},
    description={
            Discount function in PHANTOM
        }
}
\newglossaryentry{phantom-block}
{
    name={\ensuremath{A}},
    description={
            A block in the \textsc{DAG-Protocol}
        }
}
\newglossaryentry{time}
{
    name={\ensuremath{t}},
    description={
            Time
        }
}
\newglossaryentry{eulers-number}
{
    name={\ensuremath{e}},
    description={
            Euler's number
        }
}
\newglossaryentry{honest}
{
	name={\ensuremath{H}},
	description={
		Honest strategy, choosing random transactions
	}
}
\newglossaryentry{greedy}
{
	name={\ensuremath{G}},
	description={
		Greedy strategy, choosing transactions with the highest fees
	}
}
\newglossaryentry{PNE}
{
	name={PNE},
	description={
		Pure Nash Equilibrium
	}
}
\newglossaryentry{SPNE}
{
	name={SPNE},
	description={
		Subgame Perfect Nash Equilibrium
	}
}
\newglossaryentry{MNE}
{
	name={MNE},
	description={
		Mixed Nash Equilibrium
	}
}

\glssetcategoryattribute{general}{dualindex}{true}

\usepackage{subcaption}
\usepackage{floatflt}
\usepackage{amsthm}

\newcommand{\pnes}{s^{*}}

\newcommand{\phon}{P_{hon}}
\newcommand{\pmal}{P_{grd}}
\newcommand{\halfhalf}{(\frac{1}{2}, \frac{1}{2})}
\newcommand{\bracks}[2]{\left( #1, #2 \right)}

\newtheorem{claim}{Claim}

\newtheorem{corollary}{Corollary}

\definecolor{ao(english)}{rgb}{0.0, 0.5, 0.0}

\newcommand{\bazka}[4]{
	\begin{tabular}[t]{c|c|c}
		$\phon$/$\pmal$ & \gls{honest} & \gls{greedy} \\
		\hline
		\gls{honest} & (#1,#1) & (#2,#3) \\
		\hline
		\gls{greedy} & (#3,#2) & (#4,#4) \\
\end{tabular}}

\usepackage[bottom]{footmisc}

\usepackage[misc]{ifsym}

\begin{document}
	\title{Incentive Attacks on DAG-Based Blockchains with Random Transaction Selection}

	\author{
			\IEEEauthorblockN{
				Martin Pere\v{s}\'{i}ni\IEEEauthorrefmark{1}\Letter \\
				\href{mailto:iperesini@fit.vut.cz}{iperesini@fit.vut.cz}
			}
			\and
			\IEEEauthorblockN{
				Ivan Homoliak\IEEEauthorrefmark{1}\Letter \\
				\href{mailto:ihomoliak@fit.vut.cz}{ihomoliak@fit.vut.cz}
			}
			\and
			\IEEEauthorblockN{
				Federico Matteo Ben\v{c}i\'{c}\IEEEauthorrefmark{2} \\
				\href{mailto:federico-matteo.bencic@fer.hr}{federico-matteo.bencic@fer.hr}
			}
			\and		
			\IEEEauthorblockN{
				Martin Hrub\'{y}\IEEEauthorrefmark{1} \\
				\href{mailto:hruby@fit.vut.cz}{hruby@fit.vut.cz}
			}
			\and
			\IEEEauthorblockN{
				Kamil Malinka\IEEEauthorrefmark{1}\\
				\href{mailto:malinka@fit.vut.cz}{malinka@fit.vut.cz}
			}	
			\and
			\IEEEauthorblockA{
				\hspace{2.5cm} \IEEEauthorrefmark{1}Brno University of Technology,\\ 
			    \hspace{2.5cm} Faculty of Information Technology\\
			}
			\and			
			\IEEEauthorblockA{
				\IEEEauthorrefmark{2}University of Zagreb,\\
				Faculty of Electrical Engineering and Computing\\
			}
	}

	\maketitle
	\begin{abstract}
		Several blockchain consensus protocols proposed to use of Directed Acyclic Graphs (DAGs) to solve the limited processing throughput of traditional single-chain Proof-of-Work (PoW) blockchains.
		Many such protocols utilize a random transaction selection (\gls{RTS}) strategy (e.g., PHANTOM, GHOSTDAG, SPECTRE, Inclusive, and Prism) to avoid transaction duplicates across parallel blocks in DAG and thus maximize the network throughput.
		However, previous research has not rigorously examined incentive-oriented greedy behaviors when transaction selection deviates from the protocol. %
		In this work, we first perform a generic game-theoretic analysis abstracting several DAG-based blockchain protocols that use the \gls{RTS} strategy, and we prove that such a strategy does not constitute a Nash equilibrium, which is contradictory to the proof in the Inclusive paper.
		Next, we develop a blockchain simulator that extends existing open-source tools to support multiple chains and explore incentive-based deviations from the protocol.
		We perform simulations with ten miners to confirm our conclusion from the game-theoretic analysis. 
		The simulations confirm that greedy actors who do not follow the \gls{RTS} strategy can profit more than honest miners and harm the processing throughput of the protocol because duplicate transactions are included in more than one block of different chains.
		We show that this effect is indirectly proportional to the network propagation delay.
		Finally, we show that greedy miners are incentivized to form a shared mining pool to increase their profits. 
		This undermines the decentralization and degrades the design of the protocols in question.
		To further support our claims, we execute more complex experiments on a realistic Bitcoin-like network with more than 7000 nodes. %

	\end{abstract}

\section{Introduction}
\label{sec:intro}
Blockchains have become popular due to several interesting properties they offer, such as decentralization, immutability, availability, etc. 
Thanks to these properties, blockchains have been adopted in various fields, such as finance, supply chains, identity management, the Internet of Things, file systems, etc.

Nonetheless, blockchains inherently suffer from the processing throughput bottleneck, as consensus must be reached for each block within the chain.
One approach to solve this problem is to increase the block creation rate.
However, such an approach has drawbacks.
If blocks are not propagated through the network before a new block is created, a \textit{soft fork} might occur, in which two concurrent blocks reference the same parent block.
A soft fork is resolved in a short time by a fork-choice rule, and thus only one block is eventually accepted as valid.
All transactions in an \textit{orphaned} (a.k.a., stale) block are discarded.
As a result, consensus nodes that created orphaned blocks wasted their resources and did not get rewarded. %

\begin{figure}[t]
	\centering
	\includegraphics[width=0.9\linewidth]{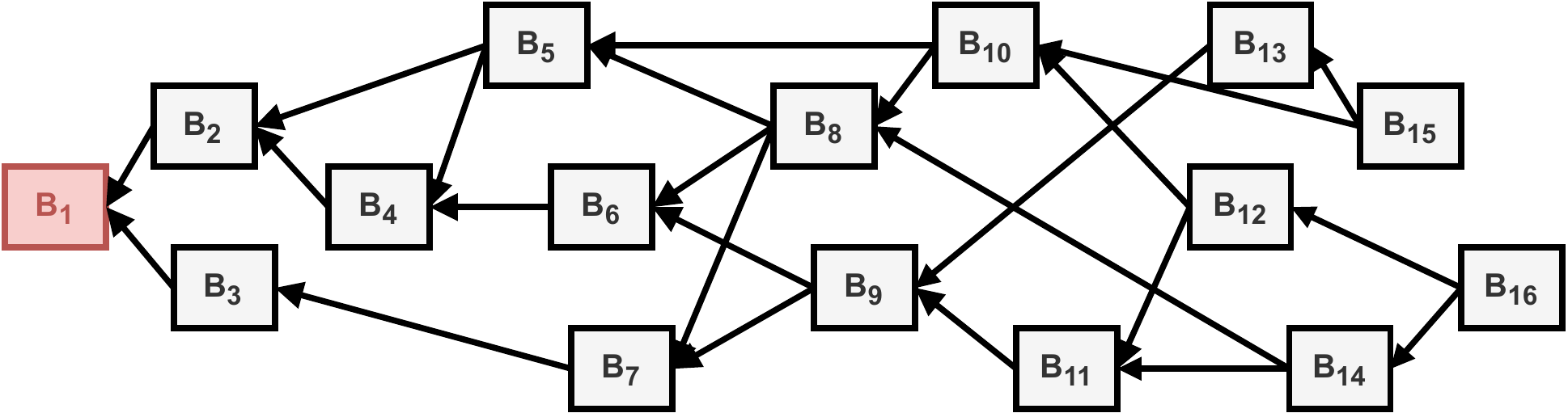}
	\caption{A structure of DAG-oriented blockchain.}
	\label{fig:dag-chain}
\end{figure}

As a response to the above issue, several proposals (e.g., Inclusive~\cite{lewenberg2015inclusive}, PHANTOM~\cite{sompolinsky2020phantom}, GHOSTDAG~\cite{sompolinsky2020phantom}, SPECTRE~\cite{sompolinsky2016spectre}) have substituted a single chaining data structure for (unstructured) \glspl{dag} (see \autoref{fig:dag-chain}), while another proposal in this direction employed structured DAG (i.e., Prism~\cite{bagaria2019prism}).
Such a structure can maintain multiple interconnected chains and thus theoretically increase processing throughput. The assumption of concerned \gls{dag}-oriented solutions is to abandon transaction selection purely based on the highest fees since this approach intuitively increases the probability that the same transaction is included in more than one block (hereafter \textit{transaction collision}).
Instead, these approaches use the random transaction selection (i.e., \gls{RTS})~\footnote{Note that \gls{RTS} involves a certain randomness in transaction selection but does not necessarily equals to uniformly random transaction selection (to be in line with the works utilizing Inclusive~\cite{lewenberg2015inclusive}, such as PHANTOM, GHOSTDAG~\cite{sompolinsky2020phantom}, SPECTRE~\cite{sompolinsky2016spectre}, as well as the implementation of GHOSTDAG called Kaspa~\cite{Kaspa}).} strategy as part of the consensus protocol to avoid transaction collisions.
Although the consequences of deviating from such a strategy might seem intuitive, no one has yet thoroughly analyzed the performance and robustness of concerned \gls{dag}-oriented approaches within an empirical study investigating incentive attacks on transaction selection.

In this work, we focus on the impact of \textbf{greedy}\footnote{Greedy actors deviate from the protocol to increase their profits.} actors in several \gls{dag}-oriented designs of consensus protocols. 
In particular, we study the situation where an attacker (or attackers) deviates from the protocol by not following the \gls{RTS} strategy that is assumed by a few \gls{dag}-oriented approaches~\cite{lewenberg2015inclusive},~\cite{sompolinsky2020phantom},~\cite{sompolinsky2020phantom},~\cite{sompolinsky2016spectre},~\cite{bagaria2019prism}.
Out of these approaches, PHANTOM~\cite{sompolinsky2020phantom}, GHOSTDAG,~\cite{sompolinsky2020phantom}, and SPECTRE~\cite{sompolinsky2016spectre} utilize \gls{RTS} that was introduced in  Inclusive~\cite{lewenberg2015inclusive} -- whose game theoretic analysis (and missing assumption about creating a mining pool) we contradict in this work.
In contrast, Prism~\cite{bagaria2019prism} does not provide any incentive-oriented analysis and thus did not show that it is resistant to any incentive attacks based on transaction selection.
Nevertheless, both lines of works employ \gls{RTS} and thus enable us to abstract their details and focus on modeling and analysis of this aspect.

We make a hypothesis stating that the attacker deviating from \gls{RTS} strategy might have two significant consequences. First, such an attacker can earn greater rewards as compared to honest participants. Second, such an attacker harms transaction throughput, as \textit{transaction collision} is increased.
We verify and prove our hypothesis in a game theoretical analysis and show that \gls{RTS} does not constitute Nash equilibrium.
Said in evolutionary terminology, a population of miners following the protocols in question is not immune against the attacker (mutant).
Next, we substantiate conclusions from game theoretical analysis by a few simulation experiments, where we focus on
an abstracted \textsc{DAG-Protocol}, inspired by existing designs.

\begin{figure}[t]
	\centering
		\vspace{0.2cm}
	\includegraphics[width=\linewidth]{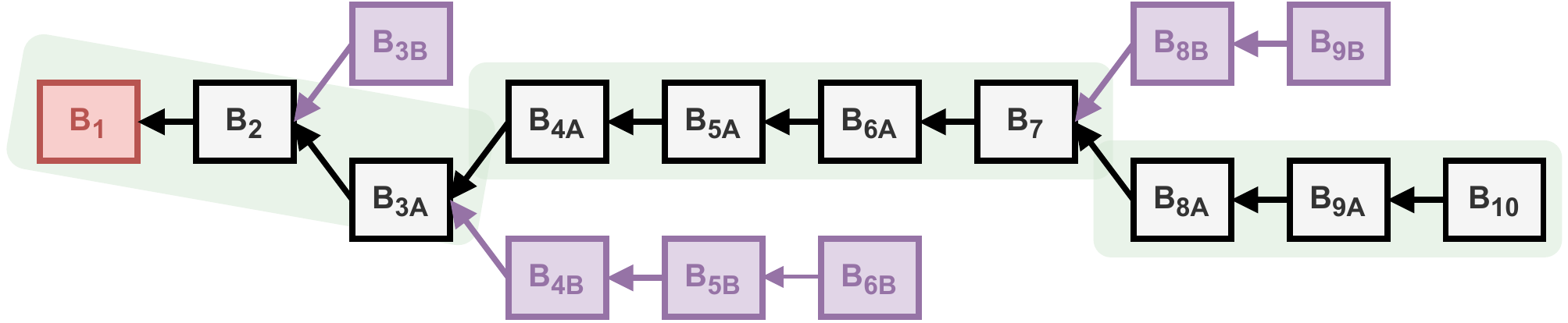}
	\caption{The longest-chain fork-choice rule with orphaned blocks depicted in purple.}\label{fig:longest-chain}
\end{figure}

\paragraph{\textbf{Contributions}}
The contributions of this work are as follows:

\begin{compactenum}

	\item We hypothesize that not following the \gls{RTS} strategy in concerned DAG-based protocols
	negatively affects the relative profit of honest miners and the effective throughput of the network.\label{contrib:1}

	\item The hypothesis is validated using the game theoretic analysis focusing on all possible scenarios involving two actors: an honest miner following \gls{RTS} and a greedy miner deviating from it.
	We conclude that the \gls{RTS} strategy does not constitute Nash equilibrium.\label{contrib:2}

	\item We build a custom simulator that extends open-source simulation tools to consider multiple chains and various incentive schemes, and thus enable us to investigate properties of concerned DAG-based protocols.\label{contrib:3}

	\item We execute experiments on an abstracted \textsc{DAG-protocol}, and
	they confirm that a greedy actor who selects transactions based on the highest fee has a significant advantage in making profits compared to honest miners following \gls{RTS}.\label{contrib:4}

	\item Next, we demonstrate by experiments that multiple greedy actors can significantly reduce the effective transaction throughput by increasing the transaction collision rate across parallel chains of DAGs.\label{contrib:5}

	\item We show that greedy actors have a significant incentive to form a mining pool to increase their relative profits, which degrades the decentralization of the concerned DAG-oriented designs. %
\end{compactenum}

\section{Background}
\label{sec:background}
We establish preliminary terms and definitions that will be used throughout this work.
The focus is put on Nakamoto's consensus that is to be optimized by DAGs.

\paragraph{Blockchain}
The blockchain is a tamper-resistant data structure in which data records (i.e., blocks) are linked using a cryptographic hash function. Each new block is agreed upon by consensus nodes running a consensus protocol.

\paragraph{\textbf{Nakamoto Consensus (NC)}}
NC~\cite{nakamoto2008bitcoin} uses a single chain to link the blocks, while \gls{pow} algorithm is used to establish consensus among nodes (i.e., miners), which is a mathematical puzzle of cryptographic zero-knowledge hash proof, where one party proves to others that it has spent a certain computational effort and thus is entitled to be a leader of the round, producing a block.
This effort represents finding a value below a  threshold (determined by the \textit{difficulty} parameter), which is computationally intensive. On the other hand, the correctness verification of the puzzle requires negligible effort.
NC is used in Bitcoin, where the order of blocks was originally determined using the longest chain fork-choice rule (see \autoref{fig:longest-chain}). 
However, this rule was later replaced in favor of the strongest chain rule (see \autoref{fig:strongest-chain}), which takes into account the accumulated difficulty of the \gls{pow} puzzle.

\begin{figure}[t]
\centering
\includegraphics[width=\linewidth]{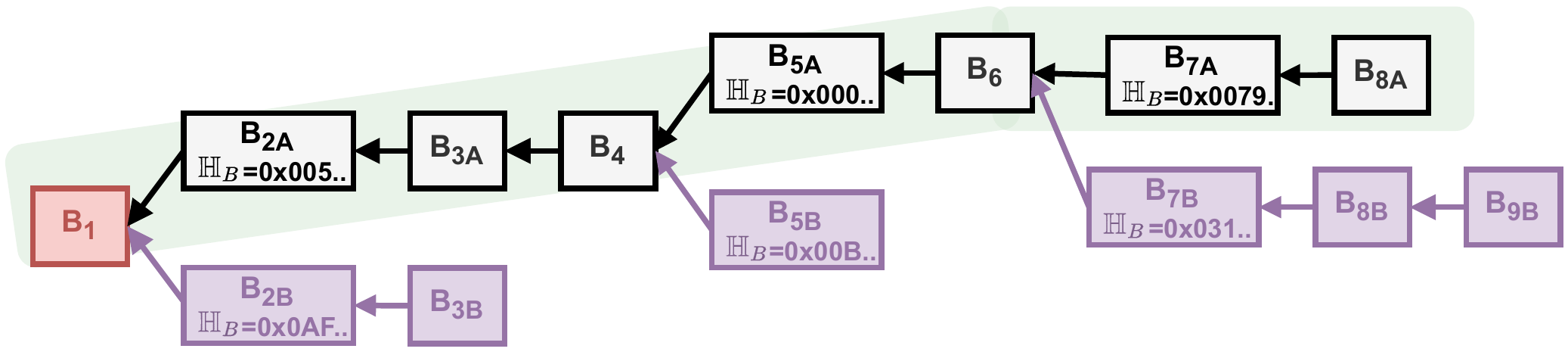}
\caption{The strongest-chain fork-choice rule with the main chain depicted in green and orphaned blocks in purple.}\label{fig:strongest-chain}
\end{figure}

\paragraph{\textbf{Fees \& Rewards}}
Miners creating new blocks are rewarded with block rewards. Block rewards refer to new crypto tokens (e.g., BTC) awarded by the blockchain network. It is assumed that miners earn profits proportionally to their consensus (i.e., mining) power. Another source of income for miners is \textit{transaction fees}, which are awarded to the miner who includes the corresponding transaction in a block. Transaction fees are paid by clients who deliberately choose the value of the fee based on the transaction's priority. To maximize profit, miners use a transaction selection mechanism that prioritizes the transactions with the highest (per Byte)\footnote{Note that since the Bitcoin block has limited capacity and transactions might have different sizes, miners consider fee normalized per Byte.} fees. %

\paragraph{\textbf{Mempool}}
A mempool is a data structure of each miner and contains transactions that can potentially be included (i.e., mined) in a block produced by a miner. A new transaction is `gossiped', i.e., sent from a client to its peers, who in turn forward the transaction to their peers, etc., until the transaction has propagated throughout the network.
Due to a network propagation delay, transactions and new blocks are not immediately propagated throughout the network. %
Therefore, the mempool might slightly vary node per node, especially at the time a new block is mined.

\paragraph{\textbf{Block Creation Time}}
In Bitcoin, there is a default block creation time \gls{block-creation-rate} set to create a new block every \(10\) minute on average. This parameter is derived directly from the network difficulty, which changes over time, and it is adjusted every \(2016\) block to fit the target value of 10 minutes (i.e., approximately every two weeks). According to Gervais et al.,~\cite{gervais2016security}, the stale block rate of Bitcoin is \(0.41\%\). Other sources~\cite{Decker13-thieee, GOBEL201623} state the values around \(0.5-1\%\), which is considered negligible. We assume that the mathematical model corresponding to \gls{block-creation-rate} of Bitcoin is an exponentially distributed random variable with the time between two consecutive blocks given by
\begin{equation}
f^{\mathbb{T}}(\gls{time})=\Lambda \gls{eulers-number}^{-\Lambda \gls{time}},
\end{equation}
where \(\Lambda = \frac{1}{\gls{block-creation-rate}}\)~\cite{bowden2018block, grunspan2020mathematics} and \gls{time} is time in seconds.%
Therefore, we model the blocks as being generated according to a Poisson process with a specified $\gls{block-creation-rate}$.

\section{Problem Definition}
\label{sec:problem}
Let there be a PoW blockchain network that uses the Nakamoto consensus and consists of honest and greedy miners, with the greedy miners holding a fraction \gls{adversarial-mining-power} of the total mining power (i.e., adversarial mining power). 
Then, we denote the network propagation delay in seconds as \gls{network-propagation-delay} and the block creation time in seconds as \gls{block-creation-rate}.
We assume that the minimum value of \gls{block-creation-rate} is constrained by \gls{network-propagation-delay} of the blockchain network.
It is well-known that Nakamoto-style blockchains generate stale blocks (a.k.a., orphan blocks).
As a result, a fraction of the mining power is wasted. The rate at which stale blocks are generated increases when \gls{block-creation-rate} is decreased, which is one of the reasons why Bitcoin maintains a high \gls{block-creation-rate} of 600s.

\paragraph{\textbf{DAG-Oriented Designs}}
Many DAG-oriented designs were proposed to allow a decrease of \gls{block-creation-rate} while utilizing stale blocks in parallel chains, which should increase the transaction throughput.
Although there are some \gls{dag}-oriented designs that do not address the problem of increasing transaction throughput (e.g., IoTA~\cite{silvano2020iota}, Nano~\cite{Lemahieu2018NanoA}, Byteball~\cite{Churyumov2016}),
we focus on the specific group of solutions addressing this problem, such as Inclusive~\cite{lewenberg2015inclusive}, GHOSTDAG, PHANTOM~\cite{sompolinsky2020phantom}, SPECTRE~\cite{sompolinsky2016spectre},  and Prism~\cite{bagaria2019prism}.
We are targeting the  \gls{RTS} strategy, which is a common property of this group of protocols.
In the \gls{RTS}, the miners do not take into account transaction fees of all included transactions; instead, they select transactions of blocks randomly -- although not necessarily uniformly at random (e.g.,~\cite{Kaspa}).
In this way, these designs aim to eliminate transaction collision within parallel blocks of the \gls{dag} structure.
Nevertheless, the interpretation of randomness in \gls{RTS} is not enforced/verified by these designs, and miners are trusted to ignore fees of all (or the majority of~\cite{Kaspa}) transactions for the common ``well-being'' of the protocol.
Contrary, miners of blockchains such as Bitcoin use a well-known transaction selection mechanism that maximizes profit by selecting all transactions of the block based on the highest fees -- we refer to this strategy as the \textit{greedy strategy} in the context of considered DAG-based protocols.

\subsection{\textbf{Assumptions}}\label{sec:attack-assumptions}
We assume a generic DAG-oriented consensus protocol using the \gls{RTS} strategy (denoted as \textsc{DAG-Protocol}).
Then, we assume that the incentive scheme of \textsc{DAG-Protocol} relies on transaction fees (but additionally might also rely on block rewards),\footnote{Note that block rewards would not change the applicability of our incentive attacks, and the constraints defined in the game theoretic model (see \autoref{sec:game-theory:model-analysis}) would remain met even with them.} and transactions are of the same size.\footnote{Note that this assumption serves only for simplification of the follow-up sections. Transactions of different sizes would require normalizing fees by the sizes of transactions to obtain an equivalent setup (i.e., a fee per Byte).}
Let us assume that the greedy miners may only choose a different transaction selection strategy to make more profit than honest miners.
Then, we assume that \textsc{DAG-Protocol} uses rewarding where the miner of the block \gls{phantom-block} gets rewarded for all unique not-yet-mined transactions in \gls{phantom-block} (while she is not rewarded for transaction duplicates mined before).

\subsection{\textbf{Identified Problems -- Incentive Attacks}}
Although the assumptions stated above might seem intuitive, there is no related work studying the impact of greedy miners deviating from the  \gls{RTS} strategy on any of the considered \textsc{DAG-protocol}s
(GHOSTDAG, PHANTOM~\cite{sompolinsky2020phantom}, SPECTRE~\cite{sompolinsky2016spectre}, Inclusive~\cite{lewenberg2015inclusive}, and Prism~\cite{bagaria2019prism})
and the effect it might have on the throughput of these protocols as well as a fair distribution of earned rewards. %
Note that we assume GHOSTDAG, PHANTOM, and SPECTRE are utilizing the  \gls{RTS} strategy that was proposed in the Inclusive protocol~\cite{sompolinsky2016spectre}, as recommended by the (partially overlapping) authors of these works -- this is further substantiated by the practical implementation of GHOSTDAG/PHANTOM called Kaspa~\cite{Kaspa}, which utilizes a variant of \gls{RTS} strategy (see \autoref{sec:dag-approaches-Kaspa}) that selects a majority portion of transactions in a block uniformly at random, while a small portion of the block capacity is seized by the transaction selected based on the highest fees.
Nevertheless, besides potentially increased transaction collision rate, even such an approach enables more greedy behavior.

\smallskip
We make a hypothesis for our incentive attacks:
\begin{hypothesis}\label{hypo:problem-definition}
	A greedy transaction selection strategy will decrease the relative profit of honest miners as well as transaction throughput in the \textsc{DAG-Protocol}.

\end{hypothesis}
\noindent
Note that the greedy transaction selection strategy deviates from the \textsc{DAG-protocol} and thus is considered adversarial.

\section{DAG-Oriented Solutions}
\label{sec:dag-oriented-solutions}

In this section, we briefly review a few \textsc{DAG-protocol}s potentially vulnerable to the incentive attacks we are investigating.

\paragraph{\textbf{Inclusive Protocol}}
\label{sec:dag-solutions:inclusive}
Lewenberg et al.~\cite{lewenberg2015inclusive} proposed a new way to structure the chain that can operate at much faster rates than Bitcoin.
The authors utilize the DAG to form blocks in a structure called the \textit{blockDAG}.
This structure is created by allowing blocks to reference multiple previous blocks, enabling less strict transaction inclusion rules that can potentially store conflicting transactions in parallel blocks due to allowing \gls{block-creation-rate} $<$ \gls{network-propagation-delay}.
This means that the system can process larger blocks faster than is possible to gossip within the bounds of \gls{network-propagation-delay}, allowing for an increase in transaction throughput.
The authors propose the protocol as a building block for other DAG-oriented protocols, and
they claim that they reduce the advantage of highly connected miners in single-chain protocols since even stale blocks (of a single-chain) are included.

Further, the authors present the key concept of \textit{randomly selecting} transactions (i.e., \gls{RTS}) to avoid collisions; however, according to their definition, the random selection does not necessarily equals to uniformly random selection. %
The authors theoretically analyze this assumption by modeling the protocol and its transaction selection as a game, in which rational miners opt to avoid collisions.
According to the authors, the game's outcome is a sequential equilibrium, where the growing fraction of greedy miners causes a decrease in their profits, which should make such a strategy less attractive
(we show this phenomenon in \autoref{fig:malicious-miners-earn-more-profit}).
However, the authors do not assume that the miners can create a mining pool, in which they can achieve significantly higher profits than honest miners (we demonstrate it in \autoref{fig:duel-miners-earn-profit}).

\paragraph{\textbf{PHANTOM}}\label{sec:dag-solutions:phantom}
The PHANTOM protocol~\cite{sompolinsky2020phantom} is a generalization of the NC's longest-chain protocol.
While in NC each block contains a hash of the previous block in the chain it extends, PHANTOM organizes blocks in a \gls{dag}.
As a result, each block may contain multiple hash references to predecessors, like in Inclusive~\cite{lewenberg2015inclusive} that is the bases for PHANTOM.
The key contribution of PHANTOM is that it totally orders all blocks by solving \textit{the maximum \textit{k}-cluster SubDAG problem}, which utilizes the concept of the main chain and the distance from it.
Unlike NC which discards the blocks out of the main chain (i.e., orphan blocks), PHANTOM includes these blocks in a \gls{dag}, %
except for the attacker-created blocks that would be weakly connected to \gls{dag}. %

PHANTOM uses the \gls{RTS} strategy proposed by the (partially overlapping) authors of the Inclusive protocol.
The incentive scheme of PHANTOM revolves around rewarding all miners who include a  transaction within a new block \gls{phantom-block}, while assuming that transactions in the parallel blocks are unique and due to a DAG will not be discarded as in single-chain blockchains.
If there are some duplicate transactions, PHANTOM rewards them only once -- in the first block that includes them, which is evaluated after establishing the total ordering.
However, such an incentive scheme must be constructed with care, as sidechain blocks might also be the result of an attack.
Therefore, the reward a miner receives for publishing \gls{phantom-block} is indirectly proportional to the discretized delay at which \gls{phantom-block} was referenced by the main chain.
For this reason, the protocol defines a measure of the delay in publishing \gls{phantom-block} w.r.t. the main chain, called the \textit{gap parameter} \gls{gap-parameter}.
The value by which the reward is ``decayed'' is determined by the discount function \gls{discount-function}, where \(\gls{discount-function}(\gls{gap-parameter}(\gls{phantom-block})) \in [0,1]\) and \gls{discount-function} is weakly decreasing.\footnote{I.e., later inclusion of the side-chain block imposes lower reward.}
Finally, the miner is rewarded for including transactions in \gls{phantom-block} using the \textit{payoff function}.
In detail, the miner gets rewarded for all non-duplicate transactions contained in \gls{phantom-block}, and after \gls{discount-function} was applied to the respective transaction fees.

\paragraph{\textbf{GHOSTDAG}}
PHANTOM is considered impractical for efficient use~\cite{sompolinsky2020phantom}, because it requires the solution of an NP-hard problem (the maximum \textit{k}-cluster SubDAG problem).
Therefore, the authors of PHANTOM have developed a greedy (heuristic) algorithm to find block clusters, obtaining the GHOSTDAG protocol.
This protocol uses greedy ordering of the \gls{dag}, which has practical advantages.

\paragraph{\textbf{Kaspa}}\label{sec:dag-approaches-Kaspa}
The \gls{RTS} strategy is utilized even in the already running blockchain Kaspa~\cite{Kaspa}, which is the implementation of the GHOSTDAG protocol.
Kaspa selects transactions using a variant of the \gls{RTS} strategy, in which a small fraction of a block is dedicated to prioritized transactions with higher fees and remaining part of a block serves for transactions selected uniformly at random.
We argue that even this approach is vulnerable to our incentive attacks since the part of the block relying on uniformly random selection cannot be enforced/verified, and thus miners might still prioritize transactions with higher fees, which can consequently result in throughput problems and incentive attacks.
Nevertheless, the current Kaspa mainnet is not saturated, and its blocks usually contain only 1 to 5 transactions,\footnote{\url{https://explorer.kaspa.org/}}
not fully utilizing the concept of \gls{dag} for increased throughput.

\paragraph{\textbf{Prism}}
Prism~\cite{bagaria2019prism} is a protocol that aims to achieve a total ordering of transactions with consistency and liveness guarantees while achieving high throughput and low latency.
Prism differs from traditional single-chain blockchains since it involves a few parallel chains rather than a single chain.
It decouples transaction confirmation, validation, and proposal, whereas these processes are traditionally tightly coupled.
Prism replaces traditional blocks with (1) transaction blocks (i.e., blocks that contain transactions), (2) voter blocks (i.e., blocks that vote for proposer blocks), and (3) proposer blocks (i.e., blocks that reference transaction blocks). %
The authors of Prism recognize that blocks mined in parallel chains
might contain duplicate transactions.
To cope with this problem, they propose to randomly divide unprocessed transactions of the local mempool into multiple queues and then create blocks using transactions only from one randomly selected queue, which is a variant of \gls{RTS} strategy and thus enables incentive attacks based on greedy strategy.
However, the authors do not provide any analysis related to such incentive attacks.

\section{Game Theoretical Analysis}
\label{sec:gametheory}
In this section, we model a \textsc{DAG-protocol}\footnote{Note that we consider DAG-based designs (described in \autoref{sec:dag-oriented-solutions}) under this generic term of \textsc{DAG-protocols} to simplify the description but not to claim that all \textsc{DAG-protocols} (with \gls{RTS}) can be modeled as we do.}
as a two-player game, in which the honest player/phenotype ($\phon$) uses the \gls{RTS} strategy and the greedy player/phenotype ($\pmal$) uses the greedy transaction selection strategy.\footnote{Even though consensus protocols might contain multiple players, they might represent only one of two behavioral phenotypes, which is sufficient for us to prove the feasibility of our attack in our game theory model.}
We assume that the fees of transactions vary -- the particular variance of fees is agnostic to this analysis.
We present the game theoretical approach widely used to analyze interactions of players (i.e., consensus nodes) in the blockchain.
Several works attempted to study the outcomes of different scenarios in blockchain networks (e.g.,~\cite{2019-Ziyao-survey-game,2020-Wang-game-mining,singh2020game}) but none of them addressed the case of \textsc{DAG-protocols} and their transaction selection mechanisms.

\smallskip
\noindent We examine the following hypothesis:
\begin{hypothesis}\label{hypo:game1}
	So-called (honest) \gls{honest}-behavior with \gls{RTS} is a Subgame Perfect Nash Equilibrium (\gls{SPNE}) in an infinitely repeated \textsc{DAG-Protocol} game. 
	This was presented in Inclusive~\cite{lewenberg2015inclusive} and we will contradict it. 
\end{hypothesis}

\noindent Generally speaking, any strategic profile $\pnes$ becomes an equilibrium (\gls{SPNE}) in an infinitely repeated game $\Gamma$ if one of the following holds:
\begin{compactitem}

	\item $\pnes$ is a Pure Nash Equilibrium (\gls{PNE}) in the base (stage) $\Gamma$ game. Then, $\pnes$ is trivially a \gls{SPNE} too.

	\item There exists an incentive making the rational players to agree on $\pnes$.
	We recall so-called \textit{Folk theorem}~\cite{Gibbons1992-ug, Osborne1994} stating that any (individually) efficient profile may become a mutual agreement (a stable profile) if the players are willing to punish a player deviating from the agreement.
	Punishing is relevant only if the targeted player is \textit{farsighted} enough.
	Let $\delta \in \langle 0,1 \rangle$ denote the discount factor \cite{Osborne1994} put by the player to her future profits.
\end{compactitem}

\begin{table}[t]
	\centering{\bazka{a}{b}{c}{d}}
	\caption{The utility functions $U_{hon},U_{mal}$ in the \textit{base game.}} %
	\label{tab:game-base}
\end{table}

\smallskip\noindent
We study the trustworthiness of Hypothesis~\autoref{hypo:game1} in the following analysis.
Our goal is to find a principle that would ensure the \gls{honest}-behavior in some natural way (self-enforcing principle).
\subsection{Model of the \textsc{DAG-Protocol}} %

Let us assume a finite non-empty population of miners.
We want to distinguish between honest (\gls{honest}) and greedy (\gls{greedy}) behavior (i.e., behavioral phenotypes).

The nature of \textsc{DAG-Protocol} imposes that players receive transaction fees after a certain delay (necessary to achieve consensus on the order of blocks).
However, we can discretize the flow of transactions into atomic rounds of the  game in order to simplify our analysis.
This allows us to study the behavior of players within a well-defined time frames. %
In every round, players decide on their actions and receive payoffs consequently.
Overall, we can model the situation as \textit{a repeated game with separate discrete rounds}.
Since no round is explicitly marked as the last one, this game is repeated infinitely.
This allows us to analyze players' behavior over an extended period of time, which is essential for understanding the long-term effects of different strategies.

We model \textsc{DAG-Protocol} in the form of \textit{an infinitely repeated two players game with a base game}
\begin{equation}
	\Gamma = (\{\phon,\pmal\};\{\gls{honest},\gls{greedy}\};U_{hon},U_{mal}),
	\label{def-basegame}
\end{equation}
where $\phon$ is the player's determination to play \gls{honest} strategy and $\pmal$ the player's determination to the \gls{greedy}-behavior.
Pure strategy \gls{honest} is interpreted as the \gls{RTS}, while \gls{greedy} strategy represents picking the transactions with the highest fees.
Payoff functions are depicted in \autoref{tab:game-base},
where the profits in the strategic profiles $(\gls{honest},\gls{honest})$ and $(\gls{greedy},\gls{greedy})$ are uniformly distributed between players.
In the following, we analyze the model in five possible scenarios with generic levels $a, b, c, d$ of the payoffs.

\begin{table}[t]
	\subcaptionbox{Scenario 1.\label{tab:sc1}}{
		\bazka{1}{0}{2}{3}
	}
	\vspace*{0.7em}
	\subcaptionbox{Scenario 2.\label{tab:sc2}}{
		\bazka{1}{0}{3}{2}
	}
	\vspace*{0.7em}
	\subcaptionbox{Scenario 3.\label{tab:sc3}}{
		\bazka{2}{0}{3}{1}
	}
	\subcaptionbox{Scenario 4.\label{tab:sc4}}{
		\bazka{2}{1}{3}{0}
	}
	\subcaptionbox{Scenario 5.\label{tab:sc5}}{
		\bazka{1}{0.5}{1.5}{1}
	}
	\label{fig:allgames}\caption{The utility functions with assigned example values.}

\end{table}

\subsection{Analysis of the Model}\label{sec:game-theory:model-analysis}

For purposes of our analysis, lets start with the assumption that \gls{greedy}-behavior is more attractive and profitable than \gls{honest}-behavior.
Otherwise, there would be no reason to investigate Hypothesis~\autoref{hypo:game1}.
Thus, let us consider $c>a$ as the basic constraint.
We also assume $c>b$, meaning that \gls{honest}-behavior loses against \gls{greedy}-behavior in the cases of $(\gls{honest},\gls{greedy})$ and  $(\gls{greedy},\gls{honest})$ profiles.
These basic constraints yield the following scenarios:
\begin{compactitem}
	\item \textbf{Scenario 1} (\autoref{tab:sc1}): $d>c>a>b$,
	\item \textbf{Scenario 2} (\autoref{tab:sc2}): $c>d>a>b$,
	\item \textbf{Scenario 3} (\autoref{tab:sc3}): $c>a>d>b$,
	\item \textbf{Scenario 4} (\autoref{tab:sc4}): $c>a>b>d$,
	\item \textbf{Scenario 5} (\autoref{tab:sc5}): where $a=d$ and $c>a$, $c>b$.
\end{compactitem}
Note that we do not assume the case $a=b$ since the presence of $\pmal$ will drain all high-fee transactions that $\phon$ would originally obtain. 
We assign numerical utilities $\{0,1,2,3\}$ to $\{a,b,c,d\}$, respecting the constraints of scenarios. %
Note that their values are irrelevant as long as the constraints of scenarios are met.
Scenarios 1 and 2 are covered just for a sake of completeness.
If the transaction fees were to cause such game outcomes, there would be no need to trust in \gls{honest}-behavior, and the system would settle in the unique $(\gls{greedy},\gls{greedy})$ \gls{PNE}. %
The behavior of players within Scenarios 3 and 4 is more complex.
We analyze these scenarios in the following, while we present the circumstances needed for the profile $(\gls{honest},\gls{honest})$ to become a stable outcome of the system.
Scenario 5 is based on the constraint saying that the sum of all incoming transaction fees is constant in any set of rounds, therefore playing either $(\gls{greedy},\gls{greedy})$ or $(\gls{honest},\gls{honest})$ should generate the same profits.

\subsubsection{\textbf{Scenario 3}}
\label{sec:scenario3}
\paragraph{\textbf{(A) Purely Non-Cooperative Interpretation}}
Scenario 3 (\autoref{tab:sc3}) represents a typical instance of so-called \textit{Prisoner's dilemma}~\cite{Osborne1994}, where \textit{cooperative profile} $(\gls{honest},\gls{honest})$ brings the highest social outcome; however, such a profile is unstable because each player does better if she deviates by playing $\gls{greedy}$.

\begin{claim}
	Players choose $(\gls{greedy},\gls{greedy})$ in Scenario 3.
\end{claim}

\begin{proof}
	(Informal) Strategy $\gls{greedy}$ strictly dominates $\gls{honest}$ and thus $(\gls{greedy},\gls{greedy})$ is the unique \gls{PNE}.
\end{proof}

\begin{corollary}
	If $\phon$ wants to follow the social norm of DAG-Protocol (which is irrational though) then $\pmal$'s best response is pure $\gls{greedy}$. 
	If $\phon$ is uncertain about her determination and plays randomly in mixed behavior $(p,1-p)$, then $\pmal$'s best response is pure $\gls{greedy}$ for any $p \in \langle 0,1 \rangle$, where expected payoff from pure \gls{greedy} is superior:
	\begin{equation}
		3p + 1(1-p) > 2p + 0(1-p).
		\label{sc3-br}
	\end{equation}
\end{corollary}

\paragraph{\textbf{(B) When Some Coordination is Allowed}}
Let us introduce coordinated behavior into the game.
Stability of $(\gls{greedy},\gls{greedy})$ profile might now become possible in the context of Folk theorem if the following two conditions are fulfilled.
(1) It must be \textit{common knowledge}~\cite{Aumann} that $\phon$ adopts so called \textit{grim trigger strategy}~\cite{Osborne1994, Mailath}, i.e., she plays $\gls{honest}$ as long as $\pmal$ plays $\gls{honest}$, and once $\pmal$ deviates, then $\phon$ turns into $\gls{greedy}$-behavior forever, bringing the game into $(\gls{greedy},\gls{greedy})$ profile.
Player $\pmal$ is punished in this way.
The first condition establishes a kind of agreement (a social norm) between players in this scenario.
(2) $\pmal$'s discount factor is higher than the minimal value $\delta$:
\begin{equation}
	\delta = \frac{\overline{v}-v}{\overline{v}-\underline{v}} = \frac{c-a}{c-d} = \frac{1}{2},
	\label{deltaCond}
\end{equation}
where $v=a$ is the payoff in the agreement profile, $\overline{v}=c$ is the payoff when deviating from the agreement and $\underline{v}=d$ is the consequence of punishments.
\noindent
Therefore, a player $i$ with $\delta_i$ evaluates her future payoffs as
\begin{equation}
	\pi(s_1,s_2,...) = \sum_{t=1} \delta_i^{t-1}\cdot U_i(s_t).
\end{equation}
E.g., a player $i$ with $\delta_i=0.5$ is indifferent between receiving $100$ in payoff now or $200$ in the future.
A player with $\delta_i \rightarrow 0$ does not bother about the future, i.e., setting agreements with such a player makes no sense.

The Folk theorem states that a player $i$ with her $\delta_i > \delta$ in the current round (e.g., the 1st round) prefers to play the agreed $\gls{honest}$ because her expectation $\pi((\gls{honest},\gls{honest}),...)$ is higher than the profit from deviating the agreement and consequent punishments.
This assumption is highly theoretical in our case, and we discuss it later in more detail (see \autoref{sec:game-conclusions}).
Also, let us note that with increasing variance in transaction fees, the \gls{greedy} behavior becomes more tempting, and thus it is difficult to believe that $\pmal$'s discount factor exceeds the gap in \autoref{deltaCond}.

\subsubsection{\textbf{Scenario 4}}
\paragraph{\textbf{(A) Purely Non-Cooperative Interpretation}}
\smallskip\phantom{a} 
Scenario 4 (\autoref{tab:sc4}) an anti-coordination game~\cite{Osborne1994} instance, so the game has two \gls{PNE}s $(\gls{honest},\gls{greedy})$ \& $(\gls{greedy},\gls{honest})$, and one Mixed Nash Equilibrium (\gls{MNE}) in mixed strategic profile $\bracks{\halfhalf}{\halfhalf}$.

\begin{claim}
	The most reasonable behavior in Scenario 4 is to play $\halfhalf$ for both players.
\end{claim}

\begin{proof}
	(Informal) This situation might contain dynamic properties and vague interpretation.
	\begin{compactitem}
		\item Let us say that two honest players occur. Then, they both can play $(\gls{honest},\gls{honest})$ and gain $2$.
		\item If $\phon$ meets a true greedy player then her payoff drops to $1$ in the $(\gls{honest},\gls{greedy})$ profile.
		\item If $\phon$ is uncertain about the character of her opponent, i.e., she expects mixed behavior $\halfhalf$ from her opponent, then her expectation from playing pure \gls{honest} drops to $\frac{3}{2}$.
		The same expectation applies to playing pure \gls{greedy}.
	\end{compactitem}

	\noindent
	From $\phon$'s perspective, mixed behavior $\halfhalf$  guarantees the best stable outcome.
	If $\pmal$ expects $\halfhalf$ behavior from $\phon$, then $\pmal$'s best response is to play the same mixed behavior that establishes \gls{MNE}. %
	The players gain $(\frac{3}{2}, \frac{3}{2})$ in that \gls{MNE}, which
	is the highest expectation they can obtain.
\end{proof}

\paragraph{\textbf{(B) When Some Coordination is Allowed}}
Similarly to Scenario 3 (see \autoref{sec:scenario3}), let us assume $(\gls{honest},\gls{honest})$ agreement to be a common knowledge to both players.
Then, a punishment of the strategy \gls{greedy} played by $\phon$ should bring the game into $(\gls{greedy},\gls{honest})$ profile since \gls{honest} is $\pmal$'s best response to \gls{greedy}.
The honest player factually improves her payoff by punishing her greedy opponent.
Therefore, conclusion from Scenario 3 applies here in the same manner.

\subsubsection{\textbf{Scenario 5}}\label{sec:game-scenario-5}
\paragraph{\textbf{(A) Purely Non-Cooperative Interpretation}}
Payoff functions in Scenario 5 come from our assumptions where players should obtain equal outcomes in profiles $(\gls{honest},\gls{honest})$ and $(\gls{greedy},\gls{greedy})$.
The game is a \textit{Zero-sum game}, meaning that no player can gain more than 100\% profit, regardless of their chosen strategy since the sum of all incoming transaction fees is fixed in any set of rounds.
As a result, the total profit for all players is always "zero" (constant) if they all play \gls{honest} or \gls{greedy} strategy.
This scenario is similar to Scenario 3 (see \autoref{sec:scenario3}).
However, the concept of agreements and punishments loses any sense since $(\gls{honest},\gls{honest})$ profile is not more socially efficient than $(\gls{greedy},\gls{greedy})$.

\begin{claim}
	$(\gls{greedy},\gls{greedy})$ is the sole rational outcome of Scenario 5.
\end{claim}

\begin{proof}
	 $(\gls{greedy},\gls{greedy})$ is the unique \gls{PNE} in Scenario 5.
\end{proof}
We might appeal for the responsibility of players who should refrain from playing \gls{greedy} just because such a behavior negatively  influences the reputation/popularity of DAG-Protocol in the long term.
A dilemma of whether to utilize the shared resource in a reasonable or extensive way resembles the classical game-theoretical model called \textit{The Tragedy of Commons}~\cite{miller2003game}.
The honest player might insist on \gls{honest}, but it will only improve $\pmal$'s payoff and damage $\phon$.
That is why the game reaches stability only at $(\gls{greedy},\gls{greedy})$.

In anonymous environments, individual interests are often prioritized over collective interests.
This is because the lack of accountability makes it easier for individuals to act in their self-interest without any concerns about the welfare of the group.
Therefore, collective action and cooperation might be very difficult to achieve in anonymous settings.

\subsection{Summary of Scenarios 1-5}\label{sec:game-conclusions}
Let us view \textsc{DAG-Protocol} as a shared resource between miners, which enables them to earn some money.
Any kind of player may utilize this resource anonymously (by PoW mining).
The idea behind \textsc{DAG-Protocol} claims that rational miners will not deplete this shared resource by extensive greedy play.
If they deplete it, the resource is gone forever since the reputation of \textsc{DAG-Protocol} is destroyed.
Since players are rational, they are not supposed to let this happen.
However, this theory stands on the assumption that this resource is the only job opportunity the miners have.

The question we investigate is whether the \textsc{DAG-Protocol} is immune against greedy behavior.
Intuitively, if it is not, then the resource might be fully depleted.
Since in permissionless blockchains there is no technical way to stop the entrance of a greedy player, she might join \textsc{DAG-Protocol}.
If the greedy behavior offers a better payoff (even temporary) then greedy miners might parasite on \textsc{DAG-Protocol}.
Let us summarize our findings regarding the immunity of \textsc{DAG-Protocol} against greedy behavior.

\begin{claim}
	\label{cl-notimm}
	\textsc{DAG-Protocol} is not a mechanism immune against greedy behavior.
\end{claim}

\noindent
We examined the \textsc{DAG-Protocol} using five hypothetical scenarios and found out that:
\begin{compactenum}
	\item
	The $(\gls{honest},\gls{honest})$ profile \textbf{is not} a (base game's) \gls{PNE} in Scenario 1-5.
	Contrary, the profile $(\gls{greedy},\gls{greedy})$ \textbf{is} \gls{PNE} in Scenarios 1, 2, 3 and 5.
	In Scenario 4, the players get the best achievable expected payoff in mixed behavior $(\frac{1}{2}, \frac{1}{2})$, i.e., when choosing their transactions randomly in half of the cases and pick the most valued ones in the second half of the cases.
	Such dynamics could look like a general $\gls{honest}$-behavior of \textsc{DAG-Protocol}.
	However, it is not, because the probability of $(\gls{honest},\gls{honest})$ profile is only $\frac{1}{4}$.
	In $\frac{3}{4}$ of cases, there is at least one miner playing $\gls{greedy}$.

	\item
	In Scenarios 3 and 4, stability in $(\gls{honest},\gls{honest})$ profile can be achieved; however, it puts rather critical demands on the community of miners.
	They can theoretically enforce a greedy player to return into $(\gls{honest},\gls{honest})$ by punishing her in $(\gls{greedy},\gls{honest})$ or $(\gls{greedy},\gls{greedy})$ profiles.
	A rational player, who \emph{wants to stay} or \emph{must stay} in this repeated game forever, agrees upon $(\gls{honest},\gls{honest})$ if (1) her discount factor is higher than a certain gap (see \autoref{deltaCond}) and (2) punishing response from the community is guaranteed.
	The gap might also fluctuate depending on the current distribution of the transaction fees.\footnote{A distribution of the transaction fees is not the subject of a game-theoretical analysis but empirical evaluation presented later (see \autoref{sec:evaluation}).}
	Nevertheless, the practical implementation of \textbf{punishments} has serious drawbacks:
	\begin{compactenum}

		\item Honest player can detect \gls{greedy}-behavior only theoretically.
		In practical operation, the players can only guess from their previous payoffs that there is probably someone playing \gls{greedy} in the system.

		\item Greedy player can avoid punishment when she skips successive rounds and gains by doing something else (saving costs, mining on different blockchain, etc.).
		The Folk theorem applied here does not assume that the player can escape from the punishment.

		\item Finally, the principle of punishment is to execute the \gls{greedy}-behavior, which brings us to $(\gls{greedy},\gls{greedy})$ at the end.
		There is no other more suitable tool for that.
		Basically, the honest player says "do not play \gls{greedy}, otherwise, I will play \gls{greedy} as well".
	\end{compactenum}

	\item
	Scenario 5 is based on the assumption that a Zero-sum game is the natural conclusion of \gls{pow} mining.
	Players gain equally in $(\gls{honest},\gls{honest})$ and $(\gls{greedy},\gls{greedy})$ profiles.
	The honest player risks a loss when playing \gls{honest} against the greedy player.
	This makes the $(\gls{greedy},\gls{greedy})$ profile the only stable and rational outcome of this scenario.
	Scenario~5 has a strategic character of Tragedy of the Commons, where depletion of the shared resource is inevitable.
\end{compactenum}

\begin{corollary}
We conclude that Hypothesis \autoref{hypo:game1} is not valid. The $(\gls{honest},\gls{honest})$ profile is not a \gls{PNE} in any of our scenarios. Incentives enforcing \gls{honest}-behavior are hardly feasible in the anonymous (permissionless) environment of blockchains.
A community of honest miners can follow the \textsc{DAG-Protocol} until the attacker appears. The attacker playing the \gls{greedy} strategy can parasite on the system and there is no defense against such a behavior (since greedy miners can leave the system anytime and mine elsewhere, which is not assumed in~\cite{lewenberg2015inclusive}). 
Therefore, \gls{honest} is not an \emph{evolutionary stable strategy} \cite{smith_1982}, and thus \gls{honest} does not constitute a stable equilibrium.
\end{corollary}

\section{Simulation Model}
\label{sec:verification}
We created a simulation model to conduct various experiments investigating the behavior of \textsc{DAG-Protocol} under incentive attacks related to the problems identified in \autoref{sec:problem} and thus Hypothesis~\autoref{hypo:problem-definition}.
Some experiments were designed to provide empirical evidence for %
the conclusions from \autoref{sec:gametheory}.

\subsection{\textbf{Abstraction of \textsc{DAG-Protocol}}}
For evaluation purposes, we simulated the \textsc{DAG-protocol} (with \gls{RTS}) by modeling the following aspects:
\begin{compactitem}
	\item  All blocks in DAG are deterministically ordered. 
	\item  The mining rewards consist of transaction fees only.
	\item  A fee of a particular transaction is awarded only to a miner of the block that includes the transaction as the first one in the sequence of totally ordered blocks. 
\end{compactitem}
Also, in terms of PHANTOM/GHOSTDAG terminology, we generalize and do not reduce transaction fees concerning the delay from ``appearing'' of the block until it is strongly connected to the DAG.
Hence, we utilize \gls{discount-function} = 1. %
In other words, for each block \gls{phantom-block}, the discount function does not penalize a block according to its gap parameter \(\gls{gap-parameter}(\gls{phantom-block})\), i.e. \(\gls{discount-function}(\gls{gap-parameter}(\gls{phantom-block})) = 1\).
Such a setting is optimistic for honest miners and maximizes their profits from transaction fees when following the \gls{RTS} strategy.
This abstraction enables us to model the concerned problems of considered \textsc{DAG-Protocols} (see \autoref{sec:dag-oriented-solutions}).

\subsection{(Simple) Network Topology}\label{sec:network-topology-simple}
We created a simple network topology that is convenient for proof-of-concept simulations and encompasses some important aspects of the real-world blockchain network.
In particular, we were interested in emulating the network propagation delay \gls{network-propagation-delay} to be similar to in Bitcoin (i.e., $\sim5s$ at most of the time in 2022), but using a small ring topology.
To create such a topology, we assumed that the Bitcoin network contains \(7592\) nodes, according to the snapshot of reachable Bitcoin nodes found on May 24, 2022.\footnote{\url{https://bitnodes.io/nodes/}}
In Bitcoin core, the default value of the consensus node's peers is set to  8 (i.e., the node degree).\footnote{Nevertheless, the node degree is often higher than~8 in reality~\cite{mivsic2019modeling}.} 
Therefore, the maximum number of hops that a gossiped message requires to reach all consensus nodes in the network is $\sim4.29$ (i.e., $log_8(7592)$).
Moreover, if we were to assume $2-3x$ more independent blockchain clients (that are not consensus nodes), then this number would be increased to $4.83$--$4.96$.
To model this environment, we used the ring network topology with 10 consensus nodes (see \autoref{network-topology}), which sets the maximum value of hops required to propagate a message to~$5$.
Next, we set the inter-node propagation delay $\partial \tau$ to $1s$, which fits assumed \gls{network-propagation-delay} (i.e., 5s / 5 hops = 1s).
Later, we will create more complex network topology (see \autoref{sec:complex-topology-settings}).

\begin{figure}[t]
	\centering
	\includegraphics[trim={3.2cm 0 0 0}, width=.35\linewidth]{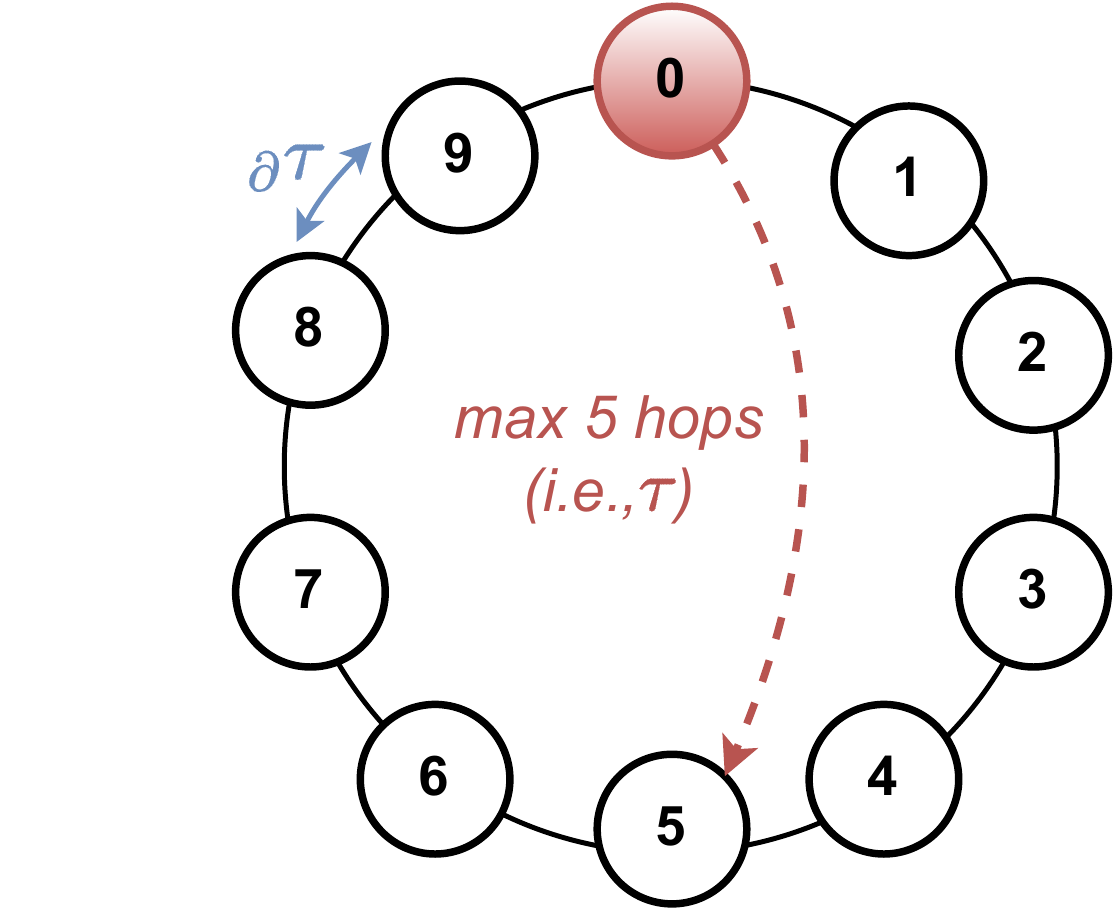}
	\caption{The simple network topology used in our simulations.}
	\label{network-topology}
\end{figure}

\subsection{Simulator}
There are many simulators~\cite{paulavivcius2021systematic} that model blockchain protocols, mainly focusing on network delays, different consensus protocols, and behaviors of specific attacks (e.g., SimBlock~\cite{Aoki2019SimBlockAB}, Blocksim~\cite{BlockSim:Alharby}, Bitcoin-Simulator~\cite{sim_bitcoin-simulator}, etc.).
However, none of these simulators was sufficient for our purposes due to missing support for multiple chains (or their abstraction) and incentive schemes assumed in \textsc{DAG-protocols}.
To verify Hypothesis \autoref{hypo:problem-definition}, we built a simulator that focuses on the mentioned problems of \textsc{DAG-protocols}.
In detail, we started with the Bitcoin mining simulator~\cite{gavinsimulator}, which is a discrete event simulator for the PoW mining on a single chain, enabling a simulation of network propagation delay within a specified network topology.

We extended this simulator to support \textsc{DAG-Protocol}s, enabling us to monitor transaction duplicity, throughput, and relative profits of miners with regard to their mining power.
The simulator is written in \verb!C++!.
The implementation utilizes the Boost library~\cite{boost} for better performance and the special structures for simulation, such as the multi-index mempool~\cite{boost2}, enabling effective management of the mempool in the case of any transaction selection strategy.\footnote{Greedy transaction strategy requires a mempool with transactions ordered by fees, while \gls{RTS} strategy requires the hash-map data structure. Therefore, it is challenging to efficiently utilize them at the same time.}

In addition, we added more simulation complexity to simulate each block, including the particular transactions (as opposed to simulating only the number of transactions in a block~\cite{gavinsimulator}). %
Most importantly, we implemented two different transaction selection strategies -- greedy and random.
For demonstration purposes, we implemented the exponential distribution of transaction fees in mempool, based on several graph cuts of fee distributions in mempool of Bitcoin from~\cite{bitcoin-mempool-stats}.\footnote{Distribution of transaction fees in mempool might change over time; however, it mostly preserves the low number of high-fee transactions in contrast to the higher number of low-fee transactions, which is common with the exponential distribution.}
Our simulator is available at \url{https://www.dropbox.com/s/vqpgqqy01qh1pcv/}.

\begin{figure}[t]
	\centering
	\includegraphics[width=0.8\linewidth]{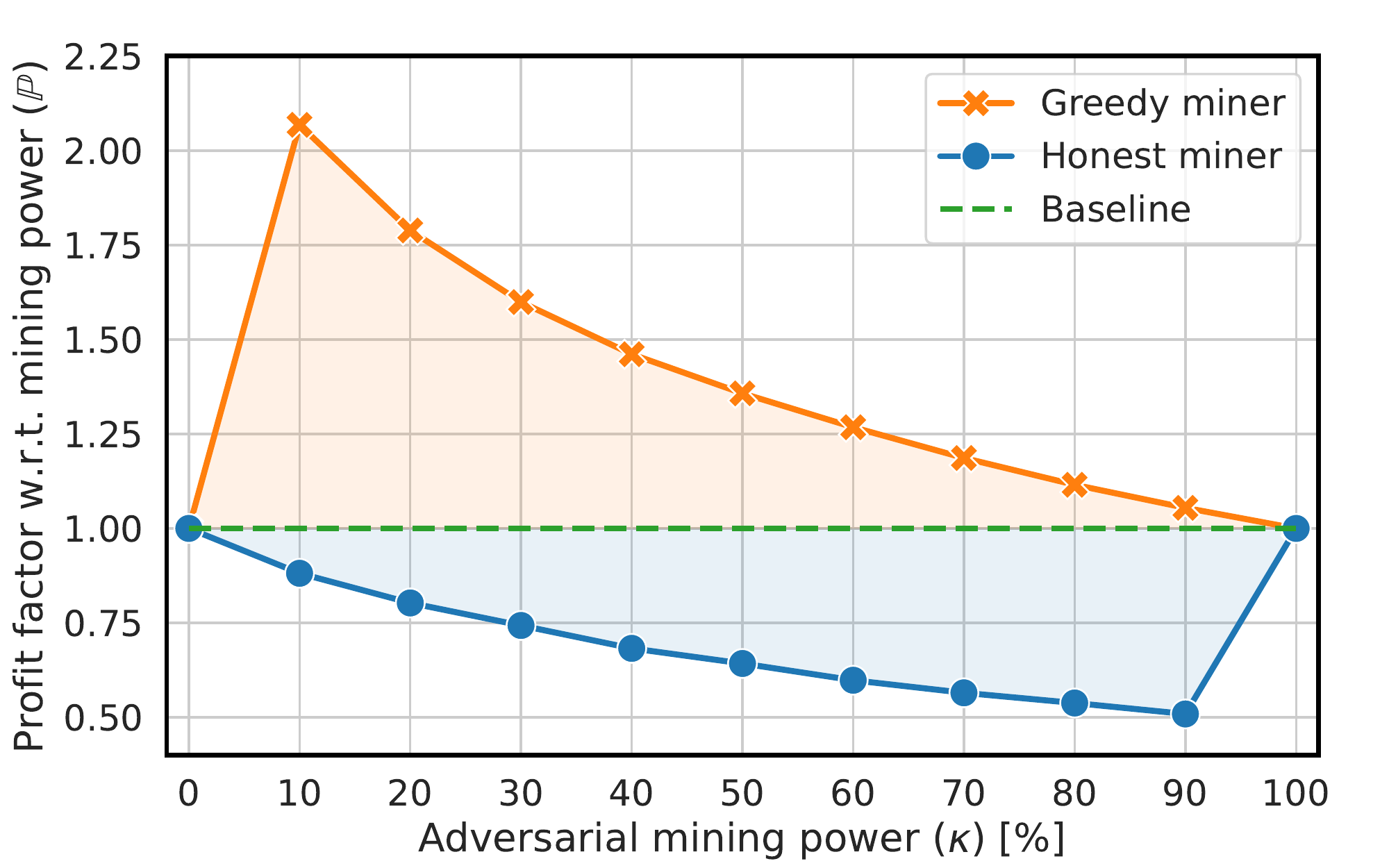}
	\caption{The profit factor $\mathbb{P}$ of an honest vs. a greedy miner with their mining powers of 100\% - \gls{adversarial-mining-power} and \gls{adversarial-mining-power}, respectively.
		The baseline shows the expected $\mathbb{P}$ of the honest miner; \(\gls{block-creation-rate}=20s\).}
	\label{fig:duel-profit}
\end{figure}

\section{Evaluation}
\label{sec:evaluation}

We designed a few experiments with our simulator, which were aimed at investigating the relative profit of greedy miners and transaction collision rate (thus throughput) to investigate Hypothesis~\autoref{hypo:problem-definition}.
In all experiments, honest miners followed the \gls{RTS},
while greedy miners followed the greedy strategy.
Unless stated otherwise, the block creation time was set to \(\gls{block-creation-rate} = 20s\). %
However, we abstracted from \gls{network-propagation-delay} of transactions and ensured that the mempools of nodes were regularly filled (i.e., every 60s) by the same set of new transactions, while the number of transactions in the mempool was always sufficient to fully satisfy the block capacity that was set to 100 transactions.
We set the size of mempool equal to 10000 transactions, and thus the ratio between these two values is similar to Bitcoin~\cite{bitcoin-mempool-stats} in common situations.
In all experiments, we executed multiple runs and consolidated their results; however, in all experiments with the simple topology, the spread was negligible, and therefore we do not depict it in graphs. 

\subsection{Experiment I}\label{sec:experiment-1}
\paragraph{\textbf{Goal}}
The goal of this experiment was to compare the relative profits earned by two miners/phenotypes in a network, corresponding to our game theoretical settings (see \autoref{sec:gametheory}). %
Thus, one miner was greedy and followed the greedy strategy, while the other one was honest and followed the \gls{RTS}.

\paragraph{\textbf{Methodology and Results}}
The ratio of total mining power between the two miners was varied with a granularity of \(10\%\), and the network consisted of 10 miners, where only the two miners had assigned the mining power.
Other miners acted as relays, emulating the maximal network delay of 5 hopes between the two miners in a duel.
The relative profits of the miners were monitored in terms of their profit factor \(\mathbb{P}\) w.r.t. their mining power.
We conducted 10 simulation runs and averaged their results (see \autoref{fig:duel-profit}).
In all simulation runs, the greedy miner earned a profit disproportionately higher than her mining power, while the honest miner's relative profit was negatively affected by the presence of the greedy miner.
We can observe that \(\mathbb{P}\)  of greedy miner was indirectly proportional to her \gls{adversarial-mining-power}, which was caused by the exponential distribution of transaction fees that contributed more significantly to the higher \(\mathbb{P}\) of a smaller miner.
In sum, this experiment showed that the profit advantage of the greedy miner aligns with the conclusions from the game theoretical model, and its Scenario 5 (see \autoref{sec:game-scenario-5}) in particular, which represents the case of \gls{adversarial-mining-power} = 50\%.
Nevertheless, our results indicate that the greedy strategy is more profitable than the \gls{RTS} for any non-zero \gls{adversarial-mining-power}.

\begin{figure}[t]
	\centering
	\includegraphics[width=0.8\linewidth]{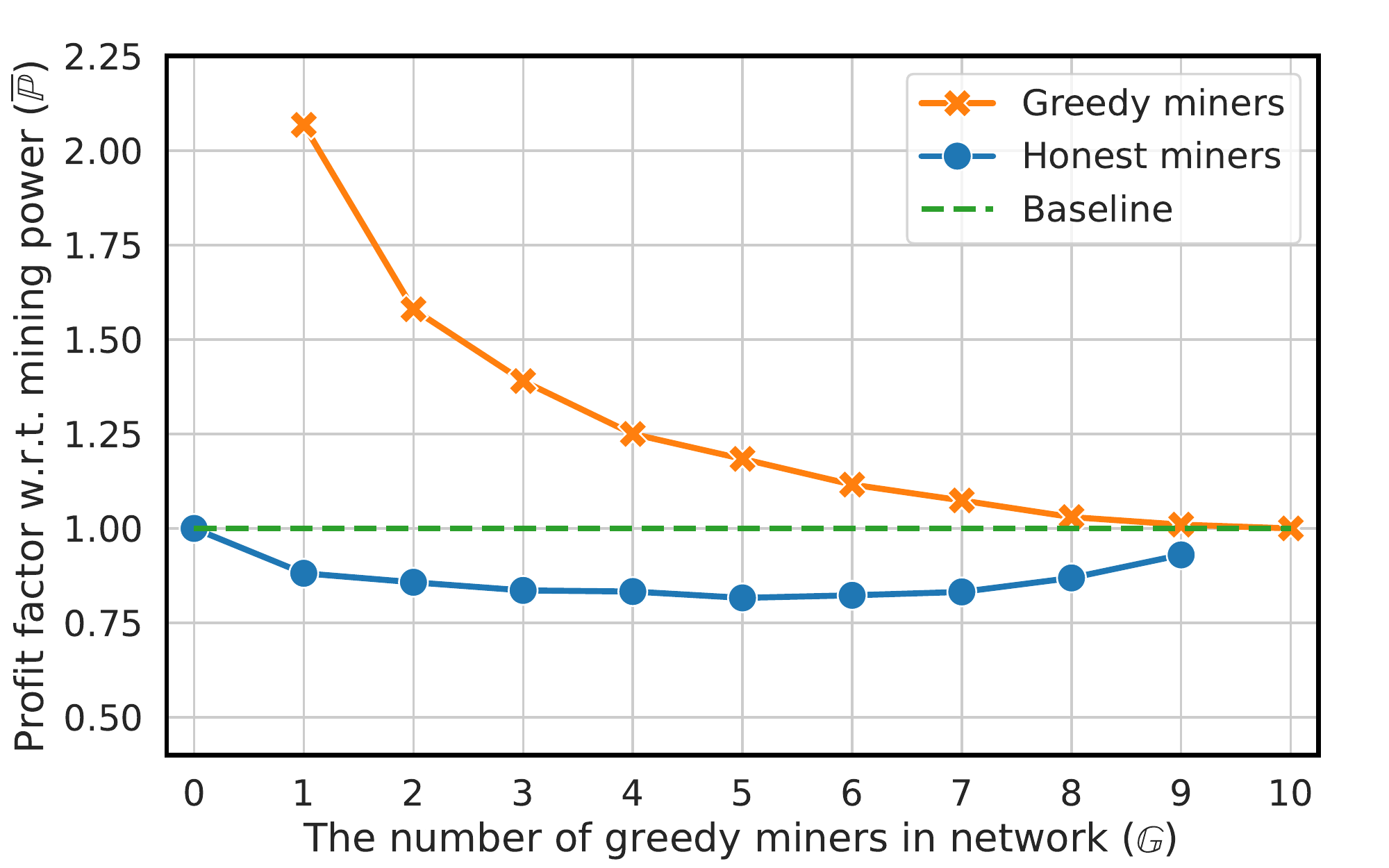}
	\caption{The averaged profit factor $\overline{\mathbb{P}}$ per honest miner and greedy miner, each with \(10\%\) of mining power.
		The number of honest miners is $10$ - \gls{cnt-malicious-miners}.
		The baseline shows the expected $\overline{ \mathbb{P}}$ of an honest miner with \(10\%\) of mining power; \(\gls{block-creation-rate}=20s\).}
	\label{fig:malicious-miners-earn-more-profit} %
\end{figure}

\begin{figure*}[t]
	\centering
	\begin{subfigure}[t]{0.325\textwidth}
		\centering
		\includegraphics[width=\textwidth]{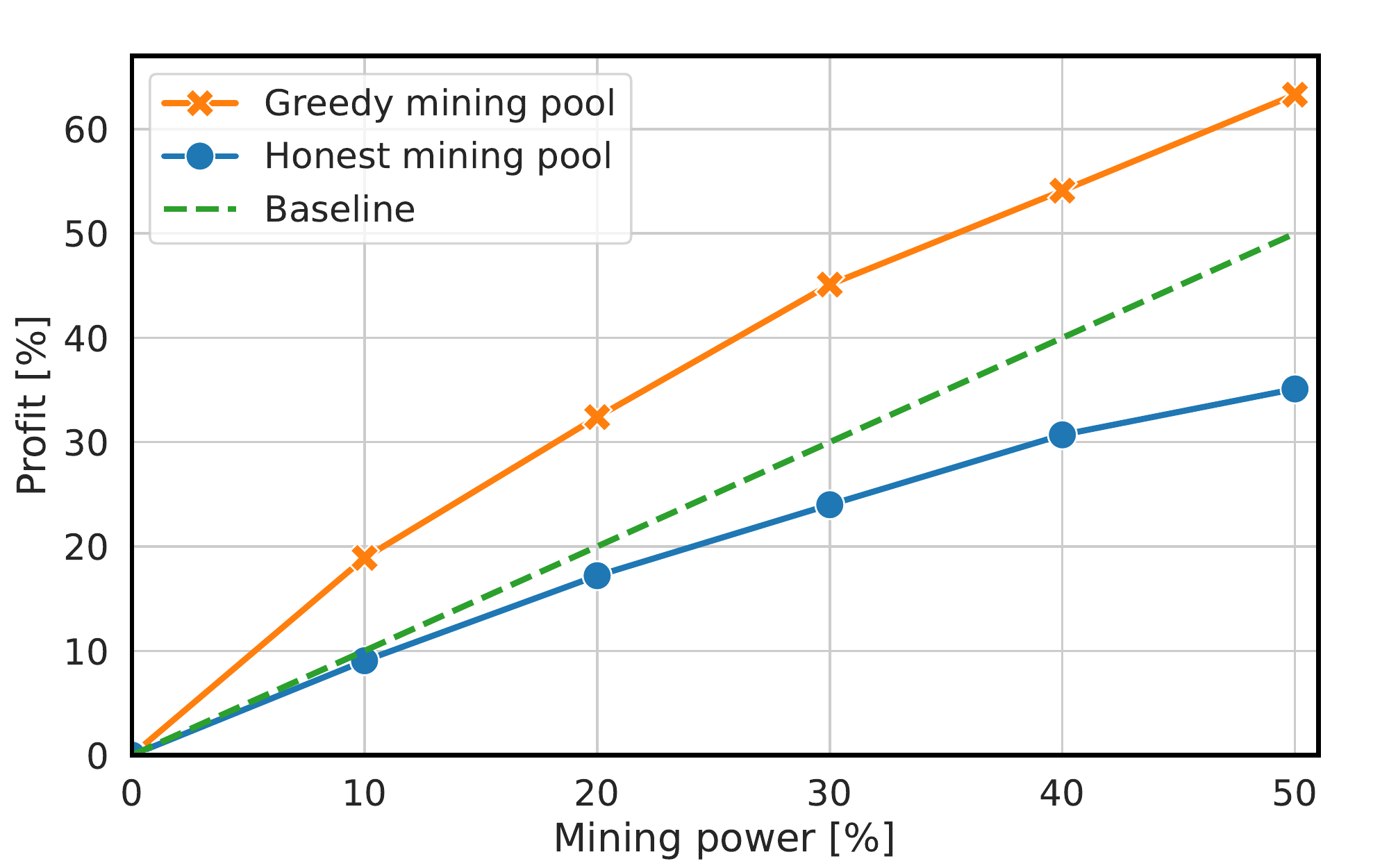}
		\caption{The relative profit of the honest pool and the greedy pool, both with equal mining power (i.e.,~\gls{adversarial-mining-power}), w.r.t. the total mining power of the network.
			The baseline shows the expected profit of the honest mining pool, and \(\gls{block-creation-rate}=20s\).}
		\label{fig:duel-miners-earn-profit}
	\end{subfigure}
	\hfill
	\begin{subfigure}[t]{0.325\textwidth}
		\centering
		\includegraphics[width=\textwidth]{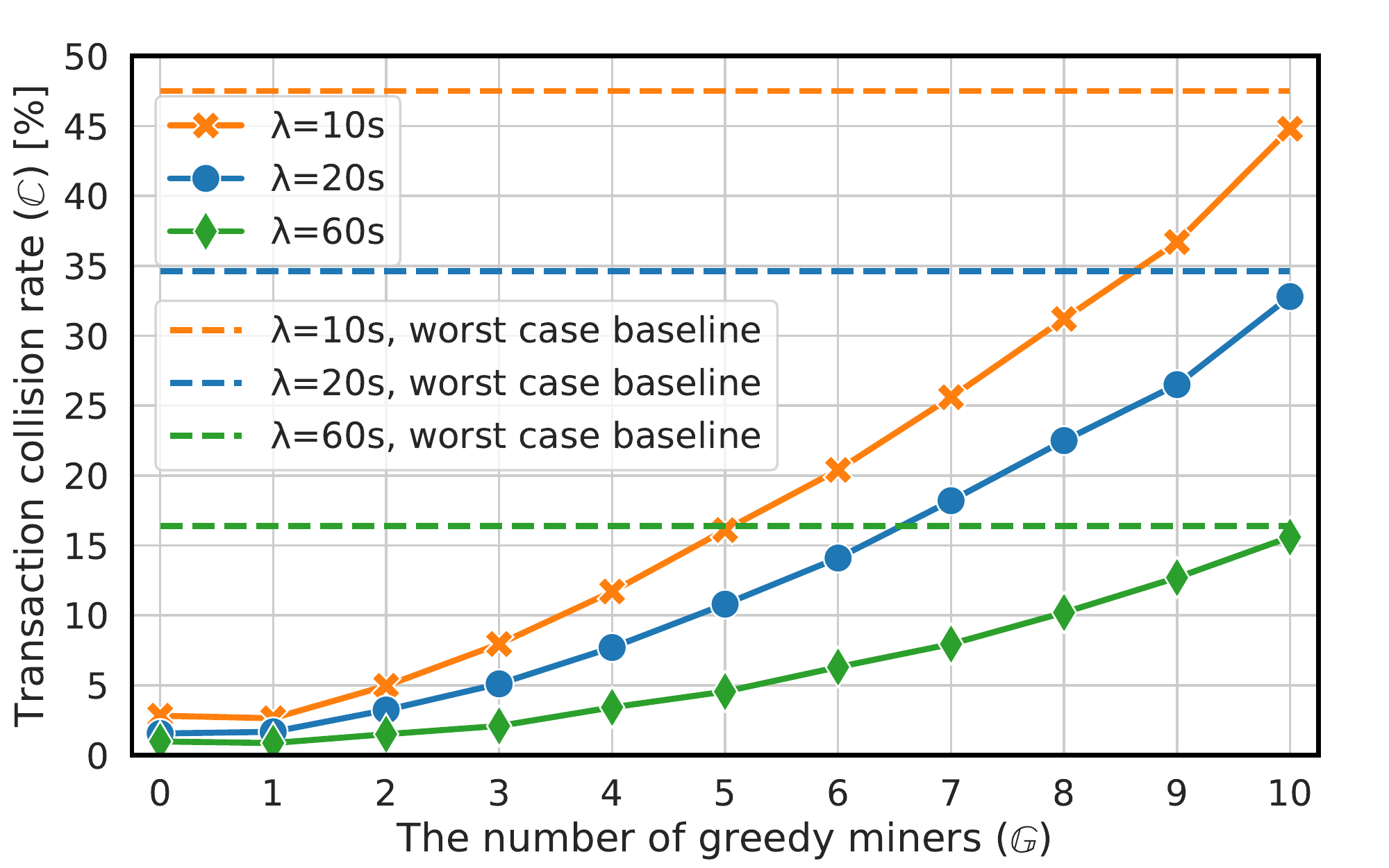}
		\caption{The transaction collision rate $\mathbb{C}$ w.r.t. \# of greedy miners \gls{cnt-malicious-miners} (each with \gls{adversarial-mining-power} = \(10\%\)), where \# of	 honest miners was $10 - \gls{cnt-malicious-miners}$ and \gls{block-creation-rate} $\in \{10s, 20s, 60s\}$.
			The worst case baseline shows $\mathbb{C}$ when all transactions are duplicates.}
		\label{fig:transaction-collision-rate-proportional-to-amount-of-malicious-miners}
	\end{subfigure}
	\hfill
	\begin{subfigure}[t]{0.325\textwidth}
		\centering
		\includegraphics[width=\textwidth]{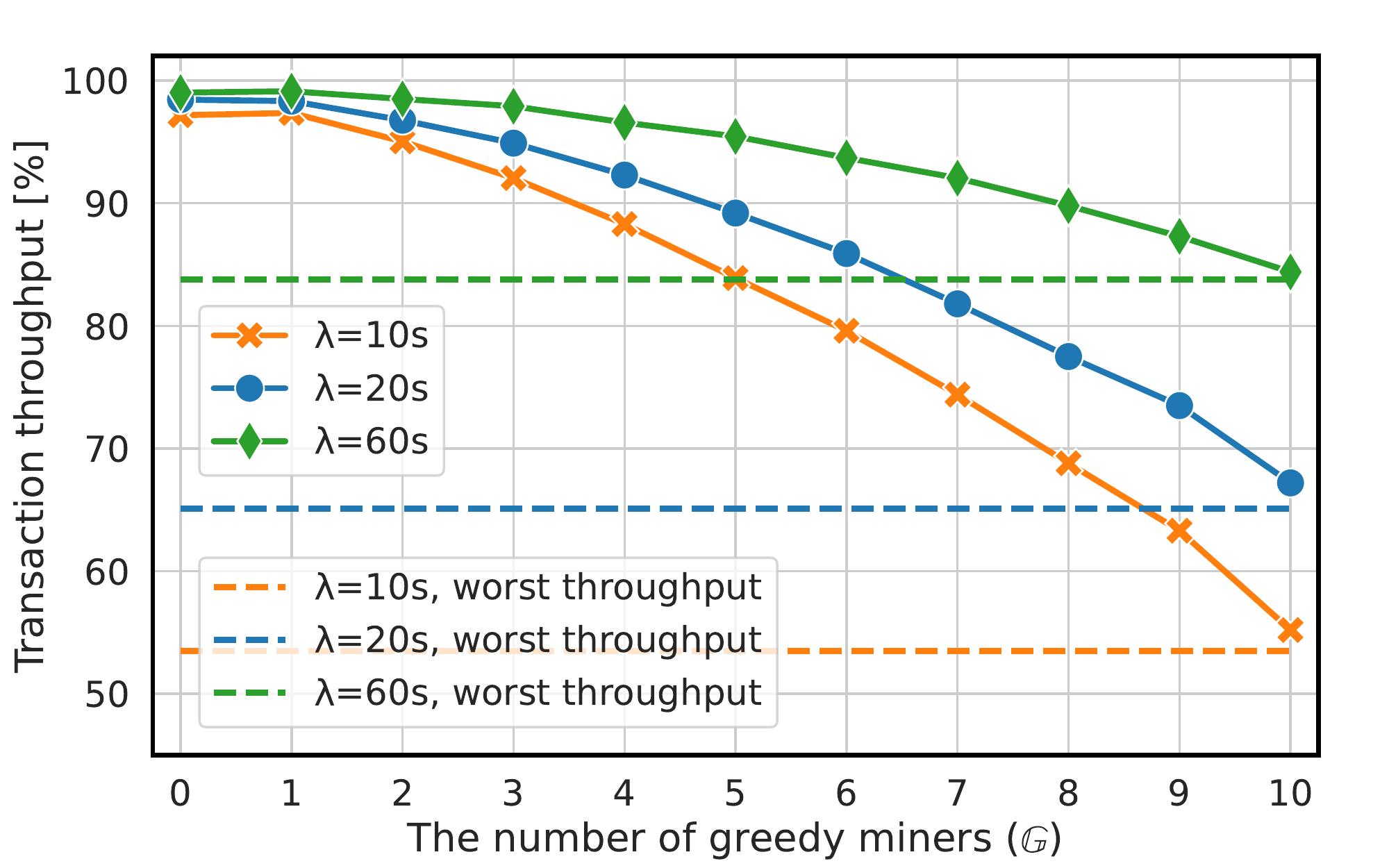}
		\caption{The throughput of the network w.r.t. \# of greedy miners \gls{cnt-malicious-miners} (each with \gls{adversarial-mining-power} = 10\%), expressed as the percentage of non-duplicate transactions mined.}
		\label{fig:throughput}
	\end{subfigure}
	\caption{Experiment III (i.e., duel of mining pools) and Experiment IV (i.e., transaction collision rate \& throughput).}
	\label{big-figure}
	\vspace{-0.5cm}
\end{figure*}

\subsection{Experiment II}\label{sec:experiment-2}
\paragraph{\textbf{Goal}}
The goal of this experiment was investigation of the relative profits of a few greedy miners following the greedy strategy in contrast to honest miners following the \gls{RTS}.

\paragraph{\textbf{Methodology and Results}}
We experimented with 10 miners, where
the number of greedy miners $\gls{cnt-malicious-miners}$ vs. the number of honest miners (i.e., 10 - \gls{cnt-malicious-miners}) was varied, and each held \(10\%\) of the total mining power.
We monitored their profit factor $\overline{\mathbb{P}}$ averaged per miner. %

We conducted 10 simulation runs and averaged their results (see \autoref{fig:malicious-miners-earn-more-profit}).
Alike in \autoref{sec:experiment-1}, we can see that greedy miners earned profit disproportionately higher than their mining power.
Similarly, this experiment showed that the profit advantage of greedy miners decreases as their number increases.
This is similar to increasing \gls{adversarial-mining-power} in a duel of two miners from \autoref{sec:experiment-1}; however, in contrast to it, $\overline{\mathbb{P}}$ of greedy miners is slightly lower with the same total \gls{adversarial-mining-power} of all greedy miners, while  $\overline{\mathbb{P}}$ of honest miners had not suffered with such a decrease.
Intuitively, this happened because multiple greedy miners increase transaction collision.
In detail, since miners are only rewarded for transactions that were first to be included in a new block, the profit for the second and later miners is lost if a duplicate transaction is included.

This observation might be seen as beneficial for the protocol as it disincentivizes multiple miners to use the greedy transaction selection strategy, which would support the sequential equilibrium from Inclusive protocol~\cite{lewenberg2015inclusive}.
However, as we mentioned in \autoref{sec:dag-solutions:inclusive}, the authors of the Inclusive protocol assume no cooperating players, which is unrealistic since miners can cooperate and create the pool to avoid collisions and thus maximize their profits (resulting in a similar outcome, as in \autoref{sec:experiment-1}).
To further investigate the profits of mining pools, we performed another experiment as follows.

\subsection{Experiment III}\label{sec:exp-2.5}
\paragraph{\textbf{Goal}}
The goal of this experiment was to investigate the relative profit of the greedy mining pool depending on its \gls{adversarial-mining-power} versus the honest mining pool with the same mining power.
It is equivalent to Scenario 5 of game theoretical analysis (see \autoref{sec:game-scenario-5}) although there is the honest rest of the network.

\paragraph{\textbf{Methodology and Results}}
We experimented with 10 miners, and out of them, we choose one greedy miner and one honest miner, both having equal mining power, while the remaining miners in the network were honest and possessed the rest of the network's mining power.
In other words, we emulated a duel of the greedy mining pool versus the honest mining pool.
We conducted 10 simulation runs and averaged their results (see \autoref{fig:duel-miners-earn-profit}).
The results demonstrate that the greedy pool's relative earned profit grows proportionally to \gls{adversarial-mining-power} as compared to the honest pool with equal mining power, supporting our conclusions from \autoref{sec:gametheory}.

\subsection{Experiment IV}\label{sec:exp-3}
\paragraph{\textbf{Goal}}
The goal of this experiment was to investigate the transaction collision rate under the occurrence of greedy miners who selected transactions using the greedy strategy.

\paragraph{\textbf{Methodology and Results}}
In contrast to the previous experiments, we considered three different values of block creation time (\gls{block-creation-rate} $ \in \{10s, 20s, 60s\}$).
We experimented with 10 miners, where
the number of greedy miners $\gls{cnt-malicious-miners}$ vs. the number of honest miners (i.e., 10 - \gls{cnt-malicious-miners}) was varied, and each held \(10\%\) of the total mining power.
For all configurations, we computed the transaction collision rate (see \autoref{fig:transaction-collision-rate-proportional-to-amount-of-malicious-miners}).
We can see that the increase of $\gls{cnt-malicious-miners}$ causes the increase in the transaction collision rate.
Note that lower \gls{block-creation-rate} has a higher impact on the collision rate, and DAG protocols are designed with the intention to have small \gls{block-creation-rate} (i.e., even smaller than \gls{network-propagation-delay}).
Consequently, the increased collision rate affected the overall throughput of the network (see \autoref{fig:throughput}, which is complementary to \autoref{fig:transaction-collision-rate-proportional-to-amount-of-malicious-miners}).

\begin{figure*}[t]
	\centering
	\begin{subfigure}[b]{0.35\textwidth}
		\centering
		\includegraphics[width=\textwidth]{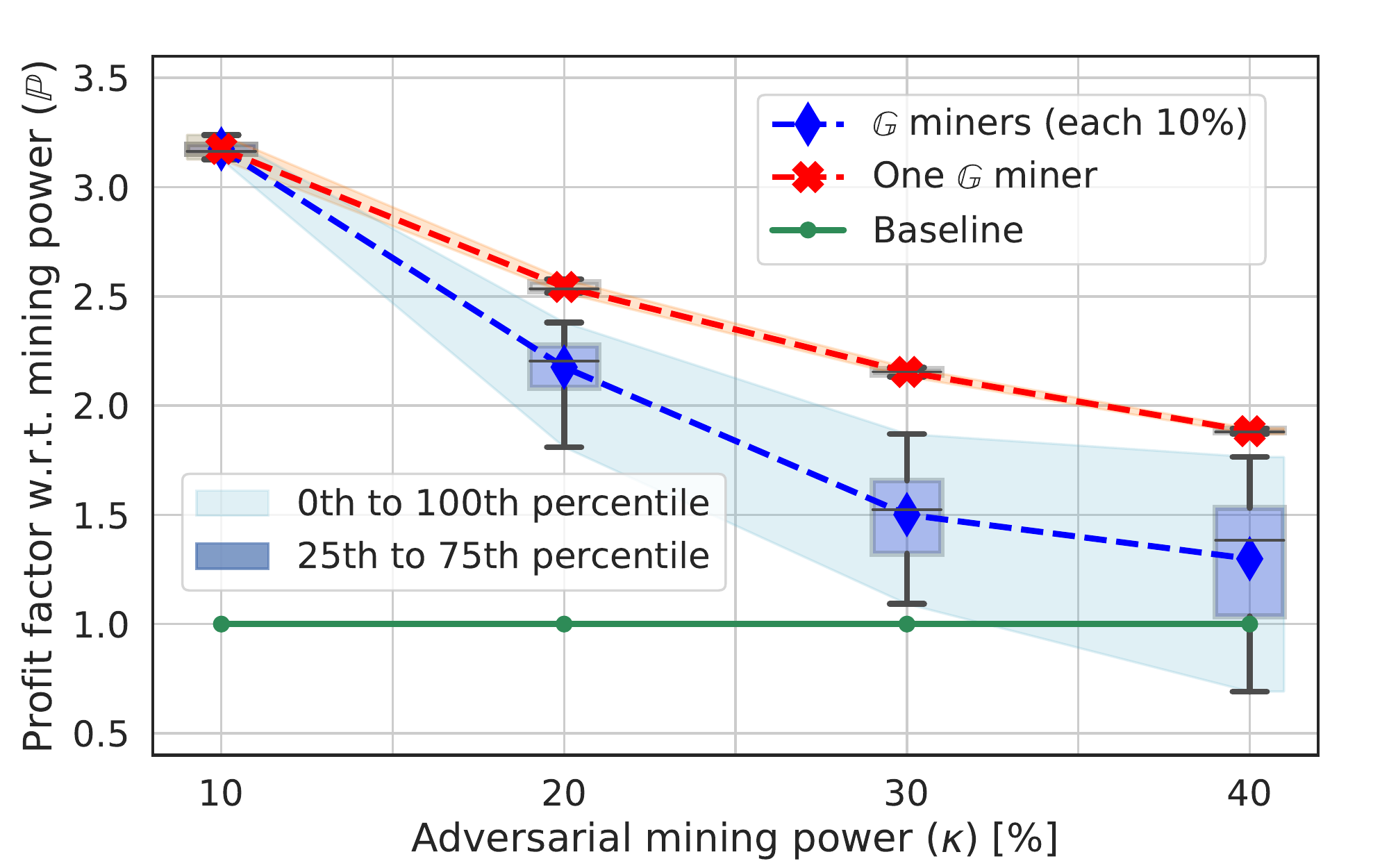}
		\vspace{-0.3cm}
		\caption{$\partial\tau$ = 5s.}
		\label{fig:complex2}
	\end{subfigure}
	\hspace{2cm}
	\begin{subfigure}[b]{0.35\textwidth}
		\centering
		\includegraphics[width=\textwidth]{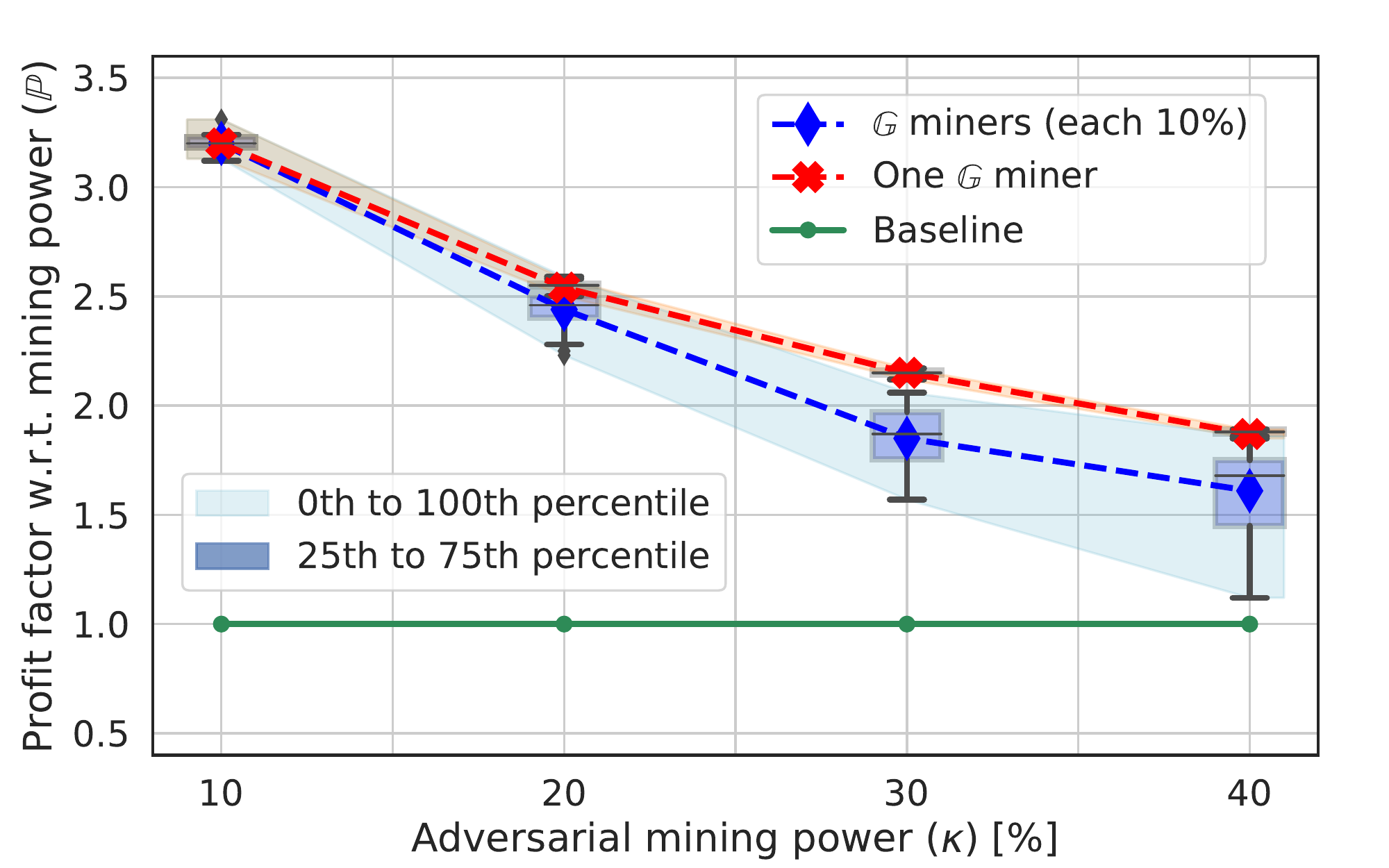}
		\vspace{-0.3cm}
		\caption{$\partial\tau$ = 0.5s.}
		\label{fig:complex4}
	\end{subfigure}
	
	\vspace{-0.4cm}
	\caption{The profit factor $\mathbb{P}$ of one \gls{adversarial-mining-power}-strong greedy miner (in {\color{red} \textbf{red}}) w.r.t. the total mining power of the network vs. the averaged $\overline{\mathbb{P}}$ of multiple greedy miners $\gls{cnt-malicious-miners} \in \langle 2,\ldots, 4 \rangle$, each with \gls{adversarial-mining-power} $= 10\%$ (in {\color{blue} \textbf{blue}}).
		The baseline represents $\mathbb{P}$ of an honest miner (in {\color{ao(english)} \textbf{green}}).}
	\label{big-figure-2}
	\vspace{-0.4cm}
	
\end{figure*}
\begin{figure*}[t]
	\vspace{-0.4cm}
	\centering
	\begin{subfigure}[b]{0.35\textwidth}
		\centering
		\includegraphics[width=\textwidth]{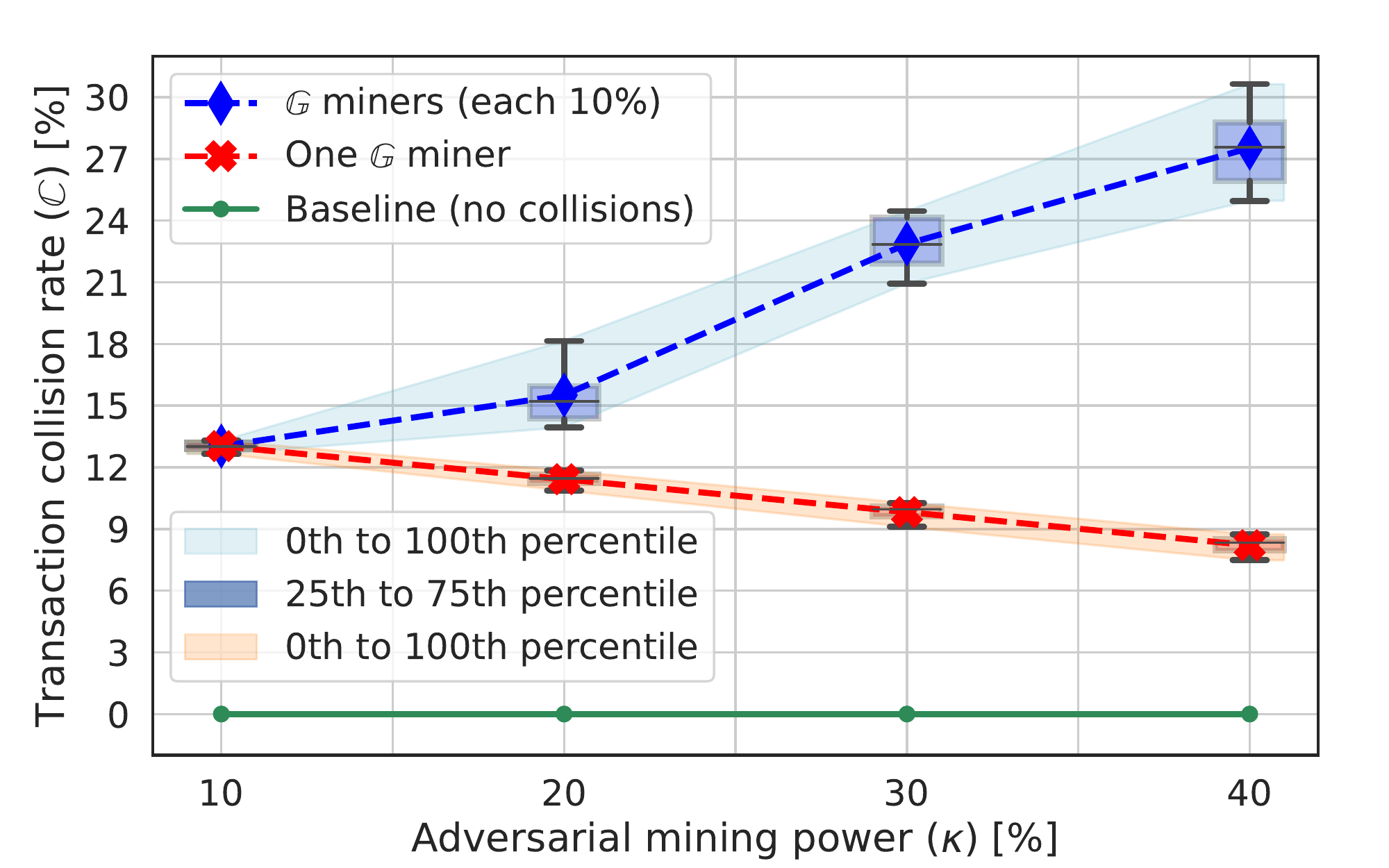}
		\vspace{-0.3cm}
		\caption{$\partial\tau$ = 5s.}
		\label{fig:collision1}
	\end{subfigure}
	\hspace{2cm}
	\begin{subfigure}[b]{0.35\textwidth}
		\centering
		\includegraphics[width=\textwidth]{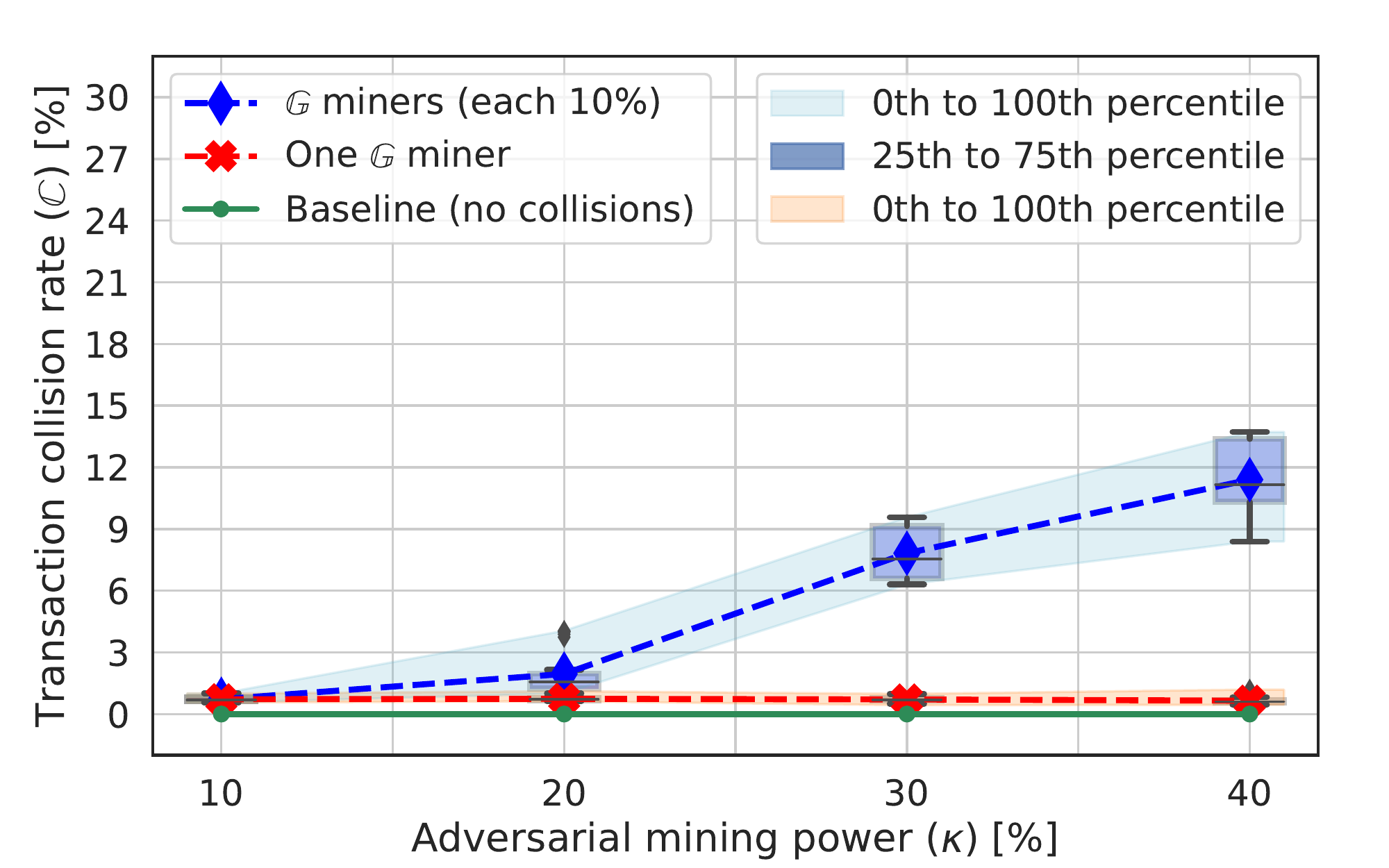}
		\vspace{-0.3cm}
		\caption{$\partial\tau$ = 0.5s.}
		\label{fig:collision2}
	\end{subfigure}
	
	\vspace{-0.4cm}
	\caption{The collision rate $\mathbb{C}$ caused by a single \gls{adversarial-mining-power}-strong greedy miner (in {\color{red} \textbf{red}}) vs. $\mathbb{C}$ caused by multiple greedy miners $\gls{cnt-malicious-miners} \in \langle 2,\ldots, 4 \rangle$, each with \gls{adversarial-mining-power} $= 10\%$ (in {\color{blue} \textbf{blue}}).
		The optimal baseline $\mathbb{C}$ of only honest miners (in {\color{ao(english)} \textbf{green}}) with no collisions.}
	\label{big-figure-3}
	\vspace{-0.4cm}
\end{figure*}

\subsection{Experiments with Complex Topology}\label{sec:complex-topology-settings}
We conducted more than 500 experiments in complex topology with \(7592\) nodes %
in various configurations (such as different connectivity and positions of greedy miners in the topology).
The generation of new transactions into mempools was made every $30$ to $120$ seconds and \gls{block-creation-rate} was set to $20$ seconds.
Since we know that \gls{network-propagation-delay} $>$ \gls{block-creation-rate} can cause a higher collision rate, we were interested in investigating this setting. 
Therefore, we distinguished two different $\partial\tau$: 0.5s and 5s, which may be considered as the lower and upper boundary (the latter meets \gls{network-propagation-delay} $>$ \gls{block-creation-rate} since 25-35s $>$ 20s).

The experiments with complex topology ran on the computation node with 20 cores and 128 GB of RAM, and they 
took 6 hours to complete on average. %
We emulated weakly and strongly connected miners by setting a different node degree -- we utilized a node degree distribution from~\cite{mivsic2019modeling} and projected it into our network by setting the weakly connected edge and a highly connected core.
We ensured the equal number of configurations with strongly and weekly connected greedy miners (which contributed to the spread in the results).

\subsection{Experiment Complex-I}\label{sec:exp-5}

\paragraph{\textbf{Goal}}
The goal of this experiment was to investigate the relative profit of one or more greedy miners w.r.t. the total mining power of the network.
We aimed at repeating the experiments from \autoref{sec:experiment-2} and \autoref{sec:exp-2.5}. %

\paragraph{\textbf{Methodology and Results}}
We compared two different scenarios.
In the first one, we experimented with \gls{adversarial-mining-power} of a single greedy miner vs. the honest rest of the network, while in the second scenario, we assumed multiple greedy miners \gls{cnt-malicious-miners} = $\in \{1, \ldots, 4\}$ (each with \gls{adversarial-mining-power} = 10\%); $\partial\tau$ was set to 5 seconds.
We can see in \autoref{fig:complex2} that a single miner is always more profitable than multiple miners, which might result in centralization by creating the greedy mining pool, as we outlined in \autoref{sec:experiment-2}.
The difference in profitability between the cases of a single and multiple miners is also significant.
If we compare these results to the simple topology, the single miner with \gls{adversarial-mining-power} = 10\% can earn 33\% of all profits in contrast to the simple topology where she can earn only 20\% of all profits.
This difference is not so significant with the higher \gls{adversarial-mining-power}.
E.g., in the case of \gls{adversarial-mining-power} = 40\%, a greedy miner on the complex topology can earn 75\% of rewards, while in the simple network it is only 55\%.
This can be caused by the high~\gls{network-propagation-delay}, favoring greedy miners that can steal new ``rich'' transactions before they can be propagated in the blocks mined by honest miners.

We re-executed the experiment with  $\partial\tau$ = 0.5s, emulating the lower boundary (i.e., real conditions).
The results are depicted in \autoref{fig:complex4}.
We can see a similar trend as with $\partial\tau$ = 5 seconds but the absolute values differ, which confirms that incentive attacks on \textsc{DAG-Protocols} are feasible even with realistic settings, not necessarily requiring \gls{network-propagation-delay} $>$ \gls{block-creation-rate}.

\subsection{Experiment Complex-II}\label{sec:exp-7}
\paragraph{\textbf{Goal}}
The goal of this experiment was to investigate the transaction collision rate $\mathbb{C}$  and the throughput of the network. %

\paragraph{\textbf{Methodology and Results}}
The methodology of the experiment is equivalent to \autoref{sec:exp-3}.
We used the setup with $\partial\tau$ adjusted to 5s and 0.5s, respectively (see \autoref{big-figure-3}).
When comparing these two settings, as one can expect, $\mathbb{C}$ is significantly smaller if $\partial\tau$ = 0.5s than in the case of $\partial\tau$ = 5s.
In the case of $\partial\tau$ = 5s, the miners have delayed information about already mined blocks (and their transactions), and thus updates of their mempools by deleting off transactions from already mined blocks are also delayed.
Therefore, the impact of incentive attacks on collision rate is more significant when \gls{network-propagation-delay} $>$ \gls{block-creation-rate} (which is a common assumption in \textsc{DAG-Protocol}s~\cite{sompolinsky2020phantom,sompolinsky2016spectre,sompolinsky2013GHOST}).
Another observation is that the lowest collision rate was achieved in the case of a single greedy miner, which decreases with increasing \gls{adversarial-mining-power} -- the miner controls the larger portion of blocks and thus decreases collisions in them.
However, this is not true with multiple such miners -- they are competing and thus negatively affect the collision rate (and their profits).
Therefore, they are incentivized in joining a mining pool to increase their profits (i.e., a single miner case).

\begin{figure}[!b]
	\vspace{0.4cm}
	\centering
	\begin{subfigure}{0.5\textwidth}
		\centering
		\includegraphics[width=.39\linewidth]{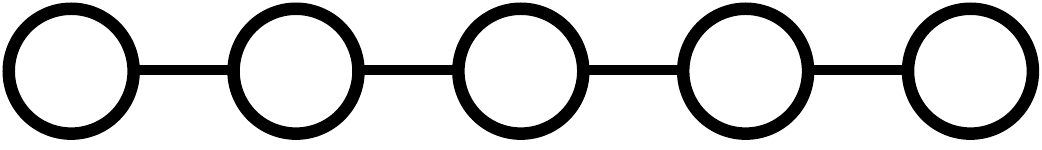}
		\caption{Line topology -- the worst case for block propagation.}\label{fig:line}
	\end{subfigure}
	\medbreak
	\begin{subfigure}{0.5\textwidth}
		\centering
		\includegraphics[width=.32\linewidth]{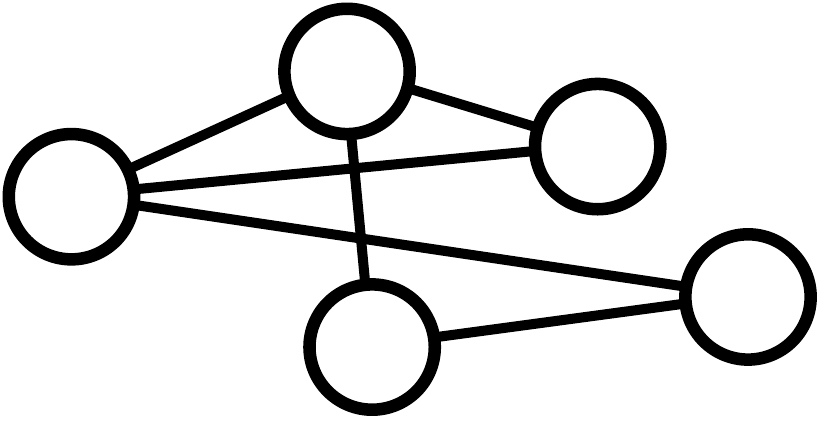}
		\caption{Common topology -- the closest to the realistic network.}\label{fig:common}
	\end{subfigure}
	\medbreak
	\begin{subfigure}{0.5\textwidth}
		\centering
		\includegraphics[width=.29\linewidth]{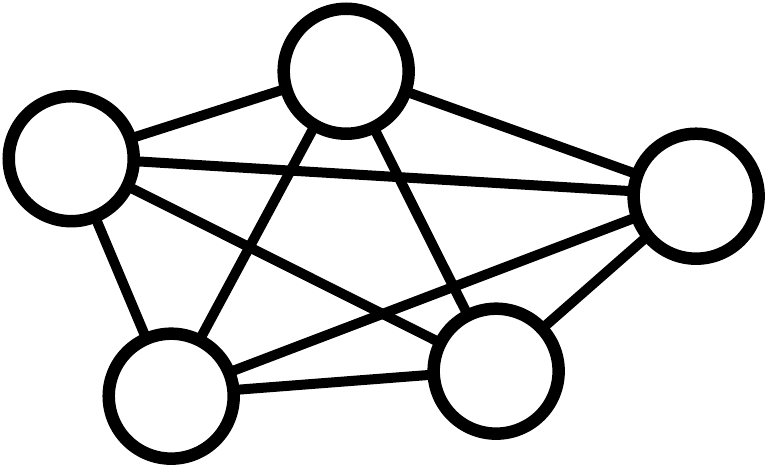}
		\caption{Fully connected topology -- the best case for block propagation.}\label{fig:fully}
	\end{subfigure}
	\caption{Various network topologies investigated in \autoref{sec:exp-8}.}
	\label{fig:topos}
\end{figure}

\subsection{Various Network Topologies}\label{sec:exp-8}
\paragraph{\textbf{Goal}}
The goal of this experiment was to investigate the effect of different network topologies on the impact of incentive attacks to \textsc{DAG-Protocol}.

\paragraph{\textbf{Methodology and Results}}
We experimented with three distinct network topologies, as shown in \autoref{fig:topos}.
Each network topology exhibited unique characteristics, which might not be realistic but enable us to investigate the effect of incentive attacks.
The line topology represented the worst-case scenario, where the gossip between any two nodes was the slowest due to the presence of the only path for block propagation. %
The common topology represented the most realistic scenario with a strongly connected core and weakly connected edge.
The fully connected topology represented the best case for the block propagation, where each message required only a single hop to be delivered.
All topologies consisted of 7000 nodes, and the $\partial\tau$ was generated using the exponential distribution reflecting an approximate $\tau$ of 5 seconds, which was fitted using the data from~\cite{misc:btcMonitoringWebsite}.
Similar to the previous experiments, a single greedy miner with \gls{adversarial-mining-power} ranging from 10\% to 40\% with a granularity of 10\% mining power was present.

Results of this experiment are depicted in \autoref{fig:agnostic}. 
We can observe that the network topology is almost an indifferent attribute.
Nevertheless, there is a slight variability, which indicates that greedy miners are still favored.
In sum, the decision to employ the greedy strategy is agnostic to network topology. %

The primary simple rationale behind considering the fully connected topology as the best scenario is its capacity to offer a comprehensive and improved understanding (overall overview) of the network's information dynamics, benefiting all participants involved.
As a result, miners give precedence to transactions that have not yet been included in blocks throughout the entire network, leveraging their complete visibility of the information.
This theoretical advantage allows for higher earnings and reduced collision rates.
This phenomenon aligns with the principle that transactions yielding profits are rewarded exclusively to the first miner to include them in a block.

\begin{figure}[t]
	\centering
	\includegraphics[width=0.9\linewidth]{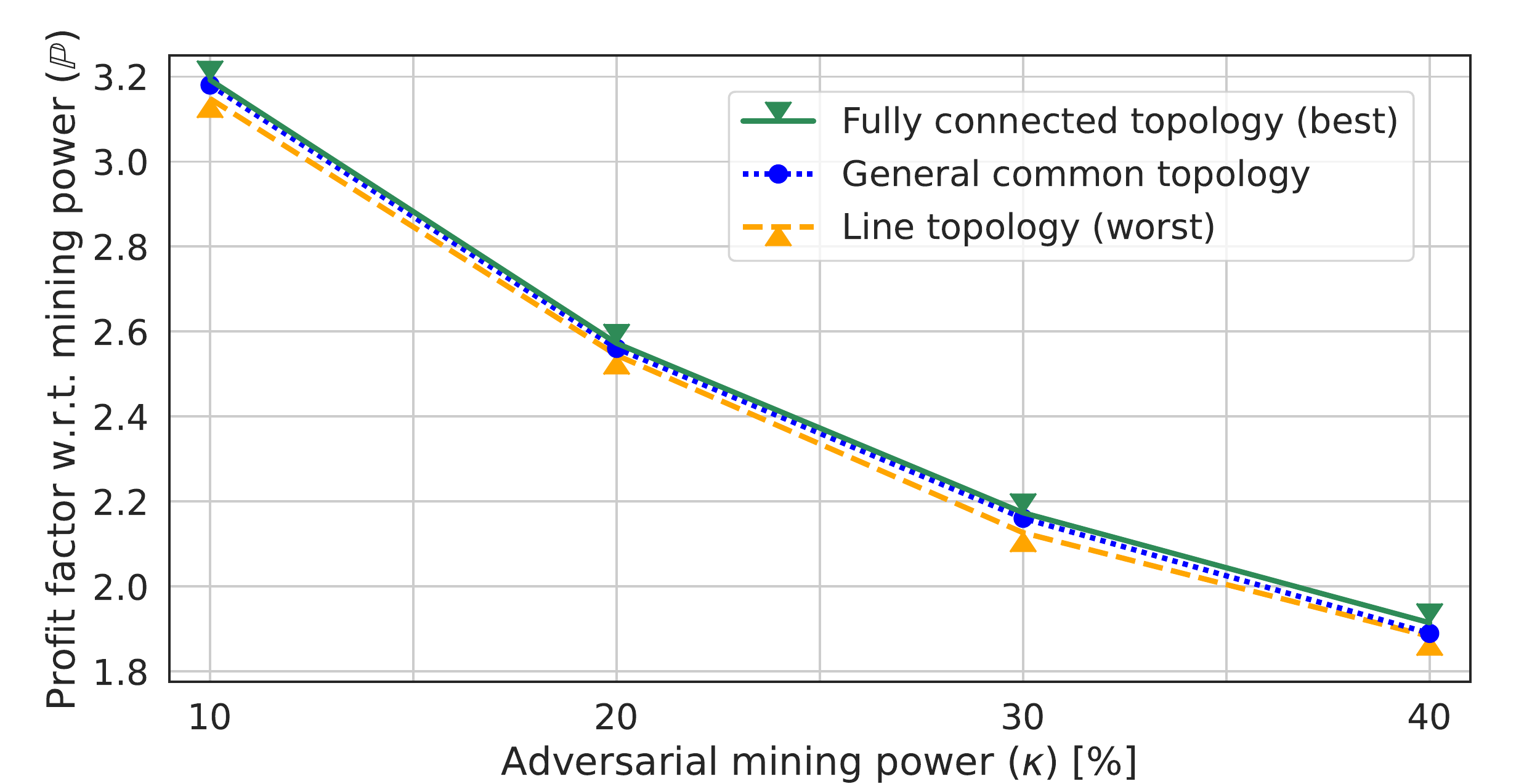}
	\caption{
		The profit factor $\mathbb{P}$ of a greedy miner with the mining power of \gls{adversarial-mining-power} = $\{10\%, \ldots, 40\%\}$.}		
		\label{fig:agnostic}
\end{figure}

\section{Countermeasures}\label{sec:countermeasures}
Our experiments supported Hypothesis \autoref{hypo:problem-definition}.
The main problem is not sufficiently enforcing the \gls{RTS}, i.e.,  verifying that transaction selection was indeed random at the protocol level.
Therefore, using the \gls{RTS} in the \textsc{DAG-Protocol} that does not enforce the interpretation of randomness will never avoid the occurrence of attackers from greedy transaction selection that increases their individual (or pooled) profits.

\paragraph{\textbf{Enforcing Interpretation of the Randomness}}
One countermeasure how to avoid arbitrary interpretation of the randomness in the \gls{RTS} is to enforce it by the consensus protocol.
An example of a DAG-based design using this approach is Sycomore~\cite{Anceaume2018Sycomore}, which utilizes the prefix of cryptographically-secure hashes of transactions as the criteria for extending a particular chain in \gls{dag}.
The PoW mining in Sycomore is further equipped with the unpredictability of a chain that the miner of a new block extends, avoiding the concentration of the mining power on ``rich'' chains.
Note that transactions are evenly spread across all chains of the DAG, which happens because prefixes of transaction hashes respect the uniform distribution -- transactions are created by clients different from miners, and clients have no incentives for biasing their transactions.

\begin{figure}[t]
	\centering
	\includegraphics[width=0.9\linewidth]{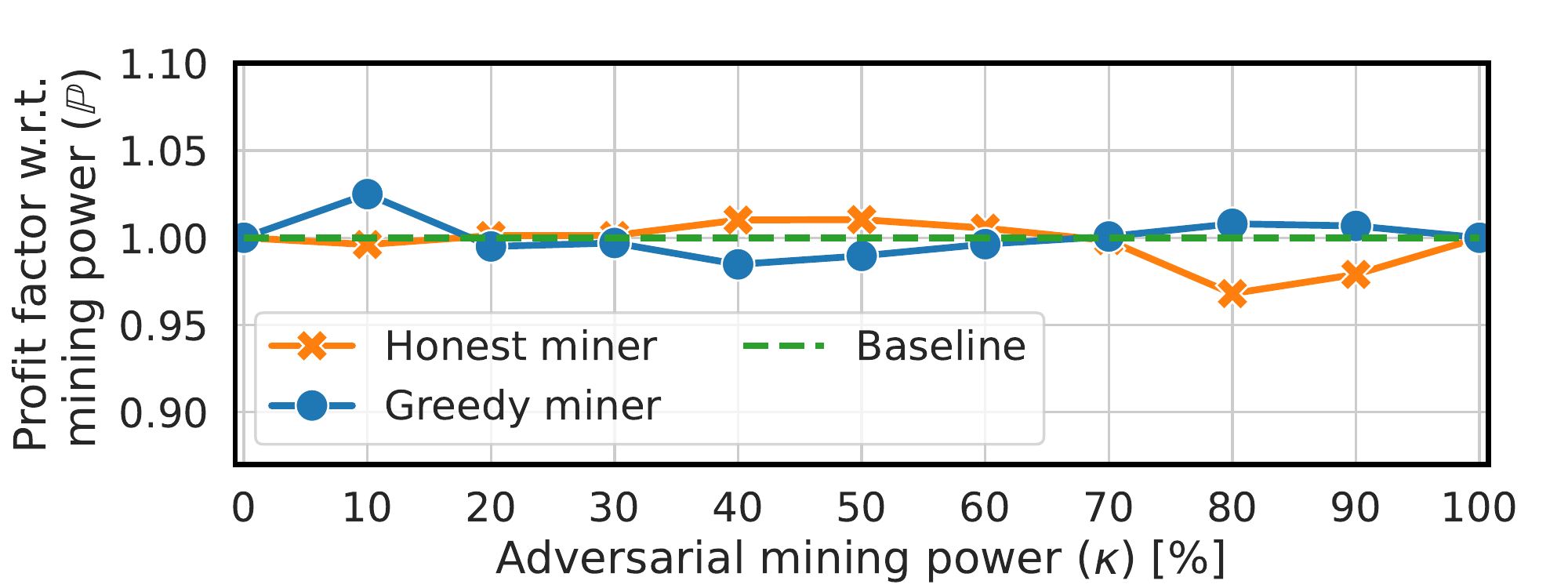}
	\caption{
		The profit factor $\mathbb{P}$ of a honest vs. a greedy miner with the mining power of 100\% - \gls{adversarial-mining-power} and \gls{adversarial-mining-power}, respectively.
		The baseline shows the expected $\mathbb{P}$ of the honest miner; \(\gls{block-creation-rate}=20s\).}\label{fig:flat-fees}
\end{figure}

\paragraph{\textbf{Fixed Transaction Fees}}
Another option how to make the \gls{RTS} viable is to employ fixed fees for all transactions as a blockchain network-adjusted parameter.
In the case of the full block capacity utilization within some period, the fixed fee parameter would be increased and vice versa in the case of not sufficiently utilized block capacity.

In contrast to the previous countermeasure, this mechanism does not enforce the interpretation of randomness while at the same time does not make incentives for greedy miners to  follow other than the \gls{RTS} strategy.
Therefore, miners using other than the \gls{RTS} would not earn extra profits -- we demonstrate it in \autoref{fig:flat-fees} and \autoref{fig:flat-fees2}, considering one honest vs. one greedy miner and one greedy vs. 9 honest miners, respectively.
Note that small deviations from the baseline are caused by the inherent simulation error that is present in the original simulator that we extended.
On the other hand, greedy miners may still cause increased transaction collision rate, and thus decreased throughput.
Therefore, we consider the fixed transaction fee option weaker than the previous one.

\section{Discussion and Future Work}
\label{sec:discussion-and-future-work}
\paragraph{\textbf{Centralization}}
In the scope of Experiment I (see \autoref{sec:experiment-1}), we demonstrated that the relative profit of greedy miners decreases as their number \gls{cnt-malicious-miners} increases. Therefore, greedy miners are incentivized to form a single mining pool, %
maximizing their relative profit.
As a negative consequence, the decentralization of the blockchain network is impacted.

\paragraph{\textbf{Throughout}}
In our simulations, we adjusted the parameters to focus on investigating potential issues related to decreased profits and general throughput (collisions), rather than maximizing the simulated protocol's throughput.
However, we argue that this had no impact on the results of our evaluation, and similar results can be achieved even with higher throughput (i.e., \gls{network-propagation-delay} $>$ \gls{block-creation-rate}).

\paragraph{\textbf{Connectivity of Miners}}
In our experiments, we used \gls{discount-function} = 1 and equally connected honest and greedy miners. However, in practice, greedy miners can be better connected since they want to include high-value transactions as the first ones, and thus profit even more. Assuming other \gls{discount-function} than 1 can result only in lower profits of (potentially) weakly connected honest miners while it does reduce the profits of greedy miners who are incentivized to be strongly connected regardless of \gls{discount-function}.

\paragraph{\textbf{Future Work}}
We plan to experiment in detail with the methods enabling \textsc{DAG-protocols} to be resilient to demonstrated attacks even in the context of Proof-of-Stake protocols. %
As the next step in theoretical analysis, it would be interesting to examine the dynamic behavior of greedy miners who adapt their strategies (switching between multiple strategies) to the current conditions of the protocol.
Last, we attributed a fee to the first miner who included a transaction in our simulations, but other schemes for distributing the fees among all miners could also be investigated.
However, such schemes would not eliminate collisions (and thus decreased throughput), thereby enabling miners still to select transactions greedily due to varying fees.

\begin{figure}[t]
	\centering
	\includegraphics[width=0.9\linewidth]{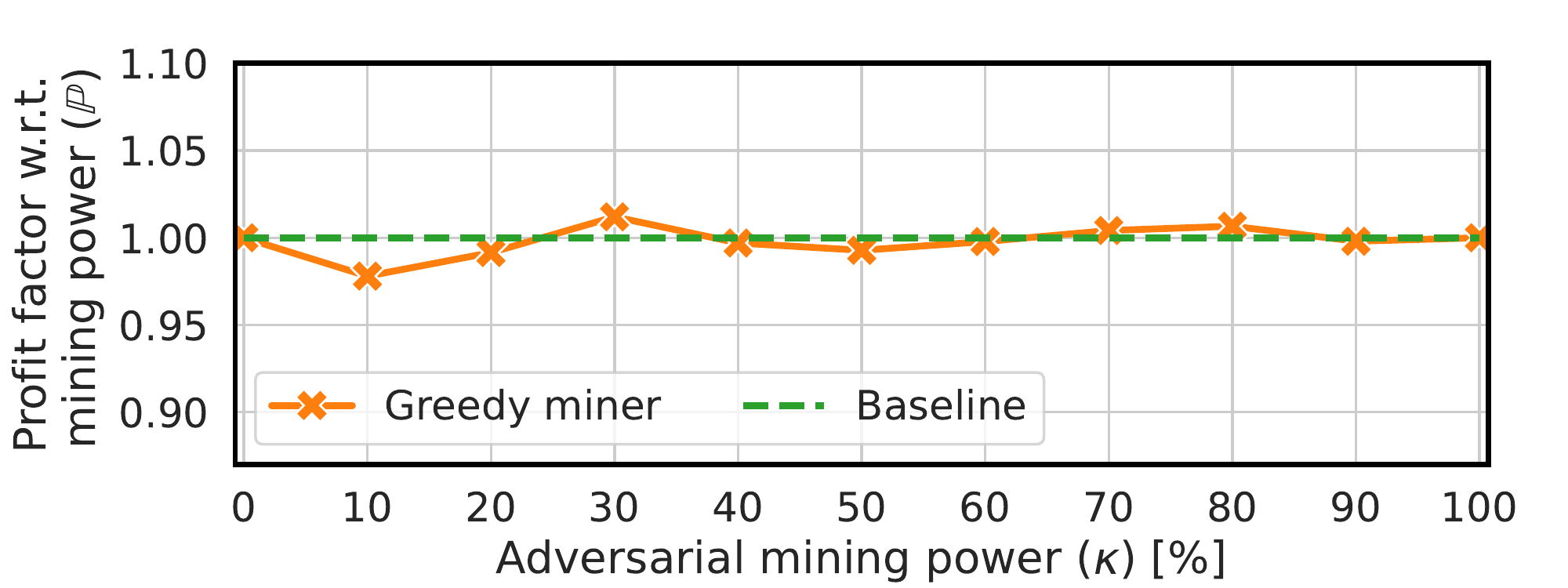}
	\caption{The averaged profit factor $\overline{\mathbb{P}}$ of a greedy miner equipped with \gls{adversarial-mining-power}.
		The rest of the network consisted of $9$ honest miners, each equipped with $\frac{100\% - \gls{adversarial-mining-power}}{9} \%$ of mining power.
		The baseline shows the expected $\overline{ \mathbb{P}}$ of an honest miner; \(\gls{block-creation-rate}=20s\).}\label{fig:flat-fees2} %
\end{figure}

\section{Related Work}
\label{sec:related}

\subsubsection*{\textbf{\gls{dag}-Based Consensus Protocols}}
The benefits of blockchain protocols come with certain trade-offs when balancing decentralization, scalability, and security.
We have already mentioned the bottleneck in Nakamoto's consensus protocol, and therefore, alternative approaches are emerging, such as \gls{dag}-based protocols.
Beside DAG-based protocols, also 2nd layer~\cite{poon2016bitcoinLightning,philips-polygon2021,homoliak2020aquareum} and sharding designs~\cite{luuElastico2016secure,ZamaniM018-rapidchain,Kokoris-KogiasJ18-omniledger}, addressing the same problems, had emerged.
However, in the current work context, we solely focus on DAG-based designs.
Wang et al.~\cite{Wang2020SoKDI} performed a detailed systematic overview of DAG-designs. They described six categories containing more than thirty \gls{dag}-based blockchain systems classified based on their characteristics and principles. They extend the commonly used classification based on the type of ledgers~\cite{Tang2020}.
GHOST~\cite{sompolinsky2015secure}, Inclusive Blockchain~\cite{lewenberg2015inclusive}, Conflux~\cite{li2018scaling}, Haootia~\cite{Tang2020}, and Byteball~\cite{Churyumov2016} represent \gls{dag} with the main chain.
Hashgraph~\cite{hashgraph2016} and Nano~\cite{Lemahieu2018NanoA} represent ledgers with parallel chains.

Nevertheless, out of these categories, DAGs with the main chain are related to our research, such as Inclusive~\cite{lewenberg2015inclusive}, SPECTRE~\cite{sompolinsky2016spectre}, PHANTOM and GHOSTDAG~\cite{sompolinsky2020phantom}.
We refer the reader to \autoref{sec:dag-oriented-solutions} for details about these protocols.
\subsubsection*{\textbf{Performance Analysis of DAGs}}
While many papers deal with the security and performance analysis of mentioned protocols, they consider neither mining strategy nor features of various incentive schemes.
Park et al.~\cite{8756973} address the performance of DAG-based blockchains and relate it to the optimization of profit. They show that the average number of parents \(n\) can influence the transaction processing time.
As a result, they propose a competitive-based transaction process system using a dynamic fee policy.

Birmpas et al.~\cite{Birmpas2020FairnessAE} propose a new general framework that captures ledger growth for a large class of DAG-based implementations to demonstrate the structural limits of DAG-based protocols. Even with honest miner behavior, fairness, and efficiency of the ledger can be affected by different connectivity levels.

One of the key technical problems of DAG-based protocols is identifying honest blocks. Wang proposes a MaxCord~\cite{MaxCord}, a framework using a different approach for honest block identification problems using graph theory. Based on the definition of the disparity measurement between blocks, they convert the problem into a maximum k-independent set problem.
Cao et al.~\cite{CAO2020480} compared the performance of three consensus mechanisms: Bitcoin (PoW), Nxt (PoS), and Tangle (\gls{dag}-PoW) in terms of parameters such as average time to generate a new block or confirmation delay and failure probability, showing how these mechanisms can affect the state of network resources or network load condition.

Sycomore~\cite{Anceaume2018Sycomore} and its extension Sycomore++~\cite{Aimen2022Sycomore++} is another DAG-oriented consensus protocol that utilizes DAGs to increase Nakamoto consensus throughput.
The protocol proposes that the chain responds to the dynamically increased number of transactions and splits them into multiple chains, thus creating \gls{dag} structure.
When the number of transactions decreases (utilization of blocks is reduced), the branches can be rejoined back into a single chain.
Transactions are evenly partitioned based on the prefix of their hash, and they are randomly inserted into their corresponding chain (branch).
The protocol does not directly suffer from our proposed attacks, although it might suffer from different problems related to double spending of transactions mined in parallel, which is however common for all DAG-oriented protocols.
\section{Conclusion}
\label{sec:conclusion}
In this work, we started with an overview of \gls{dag}-oriented consensus protocols for Proof-of-Work blockchains, which promise to increase the transaction throughput by using random transaction selection strategy.
We formulated a hypothesis that DAG protocols using the random strategy can be exploited by attackers not respecting such a strategy and instead selecting transaction based on the fees (i.e., greedy strategy), which can lead to deterioration of the overall transaction throughput.
We made a game theoretical analysis of concerned DAG-oriented protocols and concluded that the random strategy, as proposed in these protocols, does not constitute a Nash equilibrium since honest players enable the greedy player to ``parasite'' on the system.
This is contradictory result to Inclusive paper~\cite{lewenberg2015inclusive}, which does not assume that multiple greedy miners can form a mining pool. %

We conducted several experiments on simplified network topology as well as complex network using an abstracted \textsc{DAG-protocol}.
In our experiments, we analyzed the impact of greedy miners who deviated from the modeled DAG protocol by selecting transactions based on the highest fee. %
We demonstrated that greedy miners have a significant advantage over honest miners in terms of profit maximization.
Moreover, we showed that greedy miners have a detrimental impact on transaction throughput and have the incentive to form a mining pool, exacerbating the decentralization of the assumed consensus protocols.
\bibliographystyle{plain}
\bibliography{ref}

\appendix
\section{Appendix}\label{sec:appendix}

\printglossaries{}
\end{document}